\documentclass[10pt]{article}
\usepackage{epsf}
\usepackage{amsmath}

\allowdisplaybreaks

\usepackage[showframe=false]{geometry}
\usepackage{changepage}

\usepackage{epsfig}
\usepackage{amssymb}

\usepackage{amsthm}
\usepackage{setspace}
\usepackage{cite}
\usepackage{mcite}

\usepackage{algorithmic}  
\usepackage{algorithm}

\usepackage{shadow}
\usepackage{fancybox}
\usepackage{fancyhdr}

\usepackage{color}
\usepackage[usenames,dvipsnames,svgnames,table]{xcolor}

\definecolor{mypurple}{rgb}{.4,.0,.5}

\usepackage{xcolor}

\usepackage{color}

\definecolor{darkgreen}{rgb}{0, 0.4,0}
\newcommand{\dgr}[1]{\textcolor{darkgreen}{#1}}

\definecolor{purplebrown}{rgb}{0.5,0.1,0.6}

\newcommand{\bl}[1]{\textcolor{blue}{#1}}

\newcommand{\prp}[1]{\textcolor{purple}{#1}}

\definecolor{shadebrown}{rgb}{0.1,0.1,0.9}
\definecolor{lightblue}{rgb}{0.2,0,1}

\usepackage{fancybox}
\usepackage{graphicx}
\usepackage{epstopdf}
\usepackage{epsfig}
\usepackage{wrapfig}
\usepackage{subfigure}

\usepackage{xcolor}

\usepackage{tcolorbox}
\tcbuselibrary{skins}

\newtcbox{\xmyboxQ}{on line,
arc=7pt,
before upper={\rule[-3pt]{0pt}{10pt}},boxrule=1.1pt,
boxsep=0pt,left=6pt,right=6pt,top=0pt,bottom=0pt,enhanced, frame style image=blueshade.png,interior style image=goldshade.png}

\newtcbox{\xmybox}{on line,
arc=7pt,
before upper={\rule[-3pt]{0pt}{10pt}},boxrule=0pt,
boxsep=0pt,left=6pt,right=6pt,top=0pt,bottom=0pt,enhanced, coltext=blue, colback=white!10!yellow}

\newtcbox{\xmyboxa}{on line,
arc=7pt,
before upper={\rule[-3pt]{0pt}{10pt}},boxrule=0pt,
boxsep=0pt,left=6pt,right=6pt,top=0pt,bottom=0pt,enhanced, colback=white!10!yellow}

\newtcbox{\xmyboxb}{on line,
arc=7pt,
before upper={\rule[-3pt]{0pt}{10pt}},boxrule=1pt,colframe=darkgreen!100!blue,
boxsep=0pt,left=6pt,right=6pt,top=0pt,bottom=0pt,enhanced, colback=white!10!yellow}

\newtcbox{\xmyboxc}{on line,
arc=7pt,
before upper={\rule[-3pt]{0pt}{10pt}},boxrule=.7pt,colframe=blue!100!blue,
boxsep=0pt,left=6pt,right=6pt,top=0pt,bottom=0pt,enhanced, coltext=blue, colback=white!10!yellow}

\newtcbox{\xmytboxa}{on line,
arc=7pt,
before upper={\rule[-3pt]{0pt}{10pt}},boxrule=.0pt,colframe=pink!50!yellow,
boxsep=0pt,left=6pt,right=6pt,top=0pt,bottom=0pt,enhanced, coltext=white, colback=blue!40!red}

\newtcbox{\xmytboxb}{on line,
arc=7pt,
before upper={\rule[-3pt]{0pt}{10pt}},boxrule=.0pt,colframe=pink!50!yellow,
boxsep=0pt,left=6pt,right=6pt,top=0pt,bottom=0pt,enhanced, coltext=white, colback=white!40!green}

\usepackage[hyphens]{url}

\usepackage[colorlinks=true,
            linkcolor=black,
            urlcolor=blue,
            citecolor=purple]{hyperref}

\usepackage{breakurl}

\def\y{{\bf y}}
\def\v{{\bf v}}
\def\x{{\bf x}}

\def\x{{\mathbf x}}

\def\v{{\bf v}}
\def\x{{\bf x}}
\def\y{{\bf y}}
\def\z{{\bf z}}

\def\h{{\bf h}}

\def\be{\begin{equation}}
\def\ee{\end{equation}}
\def\ba{\left[\begin{array}}
\def\ea{\end{array}\right]}

\def\v{{\bf v}}
\def\x{{\bf x}}
\def\y{{\bf y}}
\def\z{{\bf z}}

\def\1{{\bf 1}}

\def\g{{\bf g}}
\def\0{{\bf 0}}

\def\erf{\mbox{erf}}
\def\erfc{\mbox{erfc}}

\def\mR{{\mathbb R}}

\def\mE{{\mathbb E}}

\newtheorem{theorem}{Theorem}

\begin{document}

\begin{singlespace}

\title {Complexity analysis of the Controlled Loosening-up (CLuP) algorithm 
}
\author{
\textsc{Mihailo Stojnic
\footnote{e-mail: {\tt flatoyer@gmail.com}} }}
\date{}
\maketitle

\centerline{{\bf Abstract}} \vspace*{0.1in}

In our companion paper \cite{Stojnicclupint19} we introduced a powerful mechanism that we referred to as the Controlled Loosening-up (CLuP) for handling MIMO ML-detection problems. It turned out that the algorithm has many remarkable features and one of them, the \emph{computational complexity}, we discuss in more details in this paper. As was explained in \cite{Stojnicclupint19}, the CLuP is an iterative procedure where each iteration amounts to solving a simple quadratic program. This clearly implies that the key contributing factor to its overall computational complexity is the number of iterations needed to achieve a required precision. As was also hinted in \cite{Stojnicclupint19}, that number seems to be fairly low and in some of the most interesting scenarios often not even larger than $10$. Here we provide a Random Duality Theory based careful analysis that indeed indicates that a very small number of iterations is sufficient to achieve an excellent performance. A solid set of results obtained through numerical experiments is presented as well and shown to be in a nice agreement with what the theoretical analysis predicts. Also, as was the case in \cite{Stojnicclupint19}, we again focus only on the core CLuP algorithm but do mention on several occasions that the concepts that we introduce here are as remarkably general as those that we introduced in \cite{Stojnicclupint19} and can be utilized in the analysis of a large number of classes of algorithms applicable in the most diverse of scientific fields. Many results in these directions we will present in several of our companion papers.

\vspace*{0.25in} \noindent {\bf Index Terms: Controlled Loosening-up (CLuP); ML - detection; MIMO systems; Algorthms; Random duality theory}.

\end{singlespace}

\section{Introduction}
\label{sec:back}

In \cite{Stojnicclupint19} we revisited the MIMO ML-detection and introduced the so-called Controlled Loosening-up (CLuP) mechanism to solve it. Since \cite{Stojnicclupint19} is the introductory paper on this subject we used it only to introduce the most basic features of the CLuP algorithm and deferred the discussion regarding many of its key advanced properties to separate papers. One of these properties, the so-called \textbf{\emph{computational complexity}}, will be the topic of the main discussion in this paper.

From the discussion presented in \cite{Stojnicclupint19} it was rather clear that the main concepts behind CLuP are very general and applicable to many different problems and algorithms used for their solving. Consequently, it was then also clear that instead of MIMO ML-detection we could have chosen quite a few other problems to introduce the main ideas behind CLuP. However, given ML-detection's importance and popularity in various scientific and engineering communities ranging from the signal processing and information theory to statistics and machine learning we selected it as a convenient choice to quickly transcendent the CLuP's basics across all of these fields. To ensure the easiness of the presentation's following and the smoothness of the connection to the already presented results in \cite{Stojnicclupint19} here we follow the same pattern and use MIMO ML detection as the underlying problem of interests. Given that we have already revisited this problem in \cite{Stojnicclupint19} we will here try to avoid repeating many of the specifics already mentioned in \cite{Stojnicclupint19} and instead focus on some of the key differences.

However, to ensure that we have the needed problem setup properly introduced we do start with a brief description of MIMO linear system. As is well known, these types of systems are modelled as:
\begin{eqnarray}\label{eq:linsys1}
\y=A\x_{sol}+\sigma\v,
\end{eqnarray}
where $A\in\mR^{m\times n}$ is the so-called system matrix (typically assumed as known at the output in the so-called coherent scenarios which we will assume here), $\x_{sol}\in\mR^n$ and $\v\in\mR^m$  are signal and noise vectors and $\sigma$ is a noise scaling factor. Finally, $\y\in\mR^m$ is the vector the the output of the system. Many practical systems can be modelled this way, with the multi-antenna systems probably being the most popular annd well-known example in the fields of information theory and wireless communications. Similarly to what we did in \cite{Stojnicclupint19}, we here also consider a statistical scenario where the elements of $\v$ and $A$ are assumed to be zero-mean unit-variance i.i.d. Gaussian random variables. Also, when it comes to the system's dimensions, we will consider the so-called \emph{linear} regime, i.e. we will consider the regime where $n$ and $m$ are large but $m=\alpha n$ where $\alpha>0$ is a real number. Of course, as is usual the case with all random duality considerations, there is really no need to restrict to this regime from the technical point of view. It is just that the writing is a bit easier and the final results are often way more elegant if one actually imposes the linear regime. Now that we have all the necessary basics, we can introduce the basic MIMO ML-detection problem as the following optimization
\begin{eqnarray}\label{eq:ml1}
\hat{\x}=\min_{\x\in{\cal X}}\|\y-A\x\|_2,
\end{eqnarray}
where, as in \cite{Stojnicclupint19}, we assume the typical binary scenario, which means that ${\cal X}=\{-\frac{1}{\sqrt{n}},\frac{1}{\sqrt{n}}\}^n$. While we will here be interested in the binary scenario we do also mention that the above ML problem is very popular in many other considerations (its LASSO/SOCP variants are among the most fundamental problems in statistics, machine learning, and compressed sensing; for more on this see, e.g. \cite{StojnicGenLasso10,CheDon95,Tibsh96,DonMalMon10,BunTsyWeg07,vandeGeer08,MeinYu09}).

Depending if one is interested in solving the optimization problem in (\ref{eq:ml1}) exactly or approximately there are quite a few excellent algorithms that have been developed throughout the literature over the years (see, e.g. \cite{GolVanLoan96Book,GroLovSch93Book,vanMaarWar00,GoeWill95}). We leave a detailed discussion about the prior work for review papers and here just briefly mention that some of the very best results that in particular relate to the type of the problems that we are interested in here can be found in e.g. \cite{StojnicBBSD05,StojnicBBSD08,FinPhoSD85,HassVik05,JalOtt05}. As in \cite{Stojnicclupint19} (and earlier in \cite{StojnicBBSD05,StojnicBBSD08}), here the goal will be to approach the \emph{\textbf{exact}} solution. To that end \cite{Stojnicclupint19} introduced the following remarkably simple iterative procedure (the above mentioned CLuP):  let $\x^{(0)}\in{\cal X}=\{-\frac{1}{\sqrt{n}},\frac{1}{\sqrt{n}}\}^n$ ($\x^{(0)}$ can be either randomly generated or designed in a specific way) and consider the following
\begin{eqnarray}
\x^{(i+1)}=\frac{\x^{(i+1,s)}}{\|\x^{(i+1,s)}\|_2} \quad \mbox{with}\quad \x^{(i+1,s)}=\mbox{arg}\min_{\x} & & -(\x^{(i)})^T\x  \nonumber \\
\mbox{subject to} & & \|\y-A\x\|_2\leq r\nonumber \\
&& \x\in \left [-\frac{1}{\sqrt{n}},\frac{1}{\sqrt{n}}\right ]^n, \label{eq:clup1}
\end{eqnarray}
where of course, as discussed in \cite{Stojnicclupint19}, $r$ is the key parameter, the so-called \emph{radius}. As was also discussed in great detail in \cite{Stojnicclupint19} the choice for $r$ is of fundamental importance for the success of the algorithm. Figure \ref{fig:fignum2} illustrates both the theoretical and simulated CLuP's performance.
\begin{figure}[htb]
\centering
\centerline{\epsfig{figure=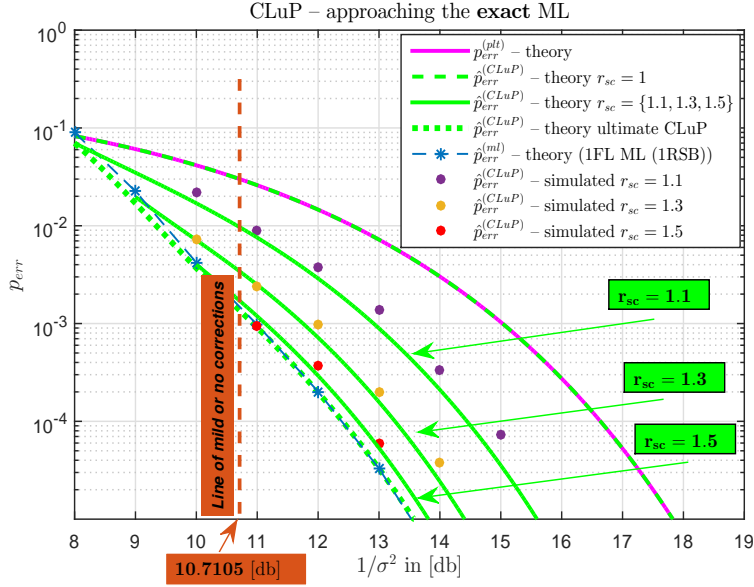,width=11.5cm,height=8cm}}
\caption{$p_{err}$ as a function of $1/\sigma^2$; $\alpha=0.8$ -- theory and simulations}
\label{fig:fignum2}
\end{figure}
Of course, all the details regarding the figure can be found in \cite{Stojnicclupint19}. Here we do mention only the key points that we view as of most relevance for the discussion that we present in this paper. Namely, as one can see from the figure, as $r$ increases from its minimal possible value $r_{plt}$ (see \cite{Stojnicclupint19} for details and a precise definition of $r_{plt}$), the CLuP's performance gets closer to the exact ML and already for $r=r_{sc}r_{plt}=1.5r_{plt}$ it is almost exactly where the predicated ML one is. We should also point out that one should note the appearance of the so-called vertical line of corrections. While we skip detailing the meaning of this line and refer instead to \cite{Stojnicclupint19}, we do mention that in this paper we will be interested in the regimes above this line where no major corrections discussed in \cite{Stojnicclupint19} are expected to take place. This is to ensure that we can focus only on one problem, \emph{\textbf{computational complexity}}, at a time and leave all others discussed in \cite{Stojnicclupint19} aside. Speaking of \emph{\textbf{computational complexity}}, we in Table \ref{tab:tabkeyclupcmpl} show how the CLuP algorithm really performs through the iterations. We selected in particular the above mentioned $r=r_{sc}r_{plt}=1.5r_{plt}$ scenario and chose a rather moderately small $n=800$. Already after $10$ iterations the CLuP's performance matches the theoretical prediction on \bl{\textbf{the fifth decimal level}} (the meaning of all relevant quantities is rather clear; we just add that $\hat{s}^{(k)}$ is the value of the CLuP's objective after the $k$-th iteration; it also goes without saying that all quantities given in the table are the expected values which due to an overwhelming concentration basically represent also the values themselves).
\begin{table}[h]
\caption{Change in $p_{err}^{(k)}$, $\hat{s}^{(k)}$, $\|\x^{(k,s)}\|_2^2$, and $(\x_{sol})^T\x^{(k,s)}$ as $k$ grows; $\alpha=0.8$; $r_{sc}=1.5$; $n=800$}\vspace{.1in}
\hspace{-0in}\centering
\footnotesize{
\begin{tabular}{||c||c|c|c|c||}\hline\hline
$k$ - iteration  & $p_{err}^{(k)}$ & $\hat{s}^{(k)}$ & $\hat{d}_2^{(k)}=\|\x^{(k,s)}\|_2^2$ & $\hat{d}_1^{(k)}=(\x_{sol})^T\x^{(k,s)}$ \\ \hline\hline
 $1$   &$0.079723  $ & $  0.17683  $ & $  0.68174  $ & $  0.71573 $ \\ \hline
 $2$   &$0.036420  $ & $  0.89707  $ & $  0.88133  $ & $  0.88448 $ \\ \hline
 $3$   &$0.014820  $ & $  0.95550  $ & $  0.94091  $ & $  0.94706 $ \\ \hline
 $4$   &$0.004763  $ & $  0.97799  $ & $  0.96899  $ & $  0.97579 $ \\ \hline
 $5$   &$0.001240  $ & $  0.98763  $ & $  0.97984  $ & $  0.98639 $ \\ \hline
 $6$   &$0.000317  $ & $  0.99090  $ & $  0.98309  $ & $  0.98946 $ \\ \hline
 $7$   &$0.000136  $ & $  0.99178  $ & $  0.98395  $ & $  0.99026 $ \\ \hline
 $8$   &$0.000083  $ & $  0.99202  $ & $  0.98419  $ & $  0.99048 $ \\ \hline
 $9$   &$0.000063  $ & $  0.99208  $ & $  0.98425  $ & $  0.99054 $ \\ \hline
$\bl{\mathbf{10}}$   &$\bl{\mathbf{0.000060}}  $ & $  \bl{\mathbf{0.99210}}  $ & $  \bl{\mathbf{0.98427}}  $ & $  \bl{\mathbf{0.99055}} $ \\ \hline\hline
\textbf{limit -- theory} & $\mathbf{0.000053}  $ & $\mathbf{0.99211}  $ & $\mathbf{0.98428}  $ & $\mathbf{0.99056}  $  \\ \hline\hline
\end{tabular}}
\label{tab:tabkeyclupcmpl}
\end{table}
In summary the mechanisms behind the CLuP algorithm introduce many remarkable properties. Two of them we will below particulary single out with respect to the MIMO ML (we leave out many other ones for the discussions regarding a host of other algorithms that we designed utilizing similar mechanisms). Namely, as one of the most fundamentally important problems at the intersection of the information theory, signal processing, statistics, machine learning and many other areas, MIMO ML has many great features that are typically of interest. Regarding the problem in (\ref{eq:ml1}) two of them are probably by far the most dominant.
\tcbset{colback=yellow!95!white,colframe=blue!95!white,fonttitle=\bfseries}
\begin{tcolorbox}[title=MIMO ML's two most fundamental theoretical/practical needs:]
\begin{enumerate}
  \item \textbf{\emph{To solve the problem in (\ref{eq:ml1}) \prp{\underline{exactly}}.}}
  \item \textbf{\emph{To solve (\ref{eq:ml1}) with the minimal needed complexity (ideally the \prp{\underline{polynomial}} one).}}
\end{enumerate}.\vspace{-.3in}
\end{tcolorbox}
This of course has been known for a long time as the heart of the matter when it comes to MIMO ML. As the above results and ultimately \cite{Stojnicclupint19} indicate CLuP manages to do rather well with respect to both of these features.
\tcbset{colback=yellow!95!white,colframe=blue!95!white,fonttitle=\bfseries}
\begin{tcolorbox}[title=CLuP's behavior regarding the MIMO ML's needs:]
\begin{enumerate}
  \item \textbf{\emph{CLuP does approach the \prp{\underline{exact}} ML.}}
  \item \textbf{\emph{Not only does the CLuP achieve the optimal performance, it does so through a \prp{\underline{fixed number}} of the simplest possible quadratic programming iterations.}}
\end{enumerate}.\vspace{-.3in}
\end{tcolorbox}

We will organize the presentation by splitting it into several parts. First we revisit the characterization of the algorithms's first iteration. Then we move to the main part which is the analysis of the second iteration. As it will soon be clear, this is the most important step of the analysis and already this very step contains all the key conceptual and technical details needed for all other steps that we will briefly touch upon afterwards. We will in parallel also provide a solid set of results obtained through simulations. As we will see, they will be of use in both, theoretical and practical aspects of the discussion. Towards the end we of course provide a few concluding remarks.

\section{CLuP -- first iteration performance analysis}
\label{sec:clupfirstit}

To analyze the computational complexity it is of course sufficient to just determine the number of iterations needed to run the algorithm to achieve the desired precision. However, we will here take a bit different and substantially more general approach. Namely, we will characterize the behavior of all critical components/parameters of the problem in each iteration. To start things off we of course first consider the first iteration. Many of the results that we present below follow as direct consequences of the main random duality theory concepts that we developed in a long line of work \cite{StojnicCSetam09,StojnicISIT2010binary,StojnicDiscPercp13,StojnicGenLasso10,StojnicGenSocp10,StojnicPrDepSocp10,StojnicRegRndDlt10}. In fact, some of the procedures will be very closely related to some of the procedures presented in \cite{Stojnicclupint19}. We will therefore try to skip repeating as many unnecessary details as possible and instead focus on the key differences between the analysis considered here and the corresponding ones considered in \cite{Stojnicclupint19}. We also do mention that not only are the analyses similar to the one that we present below of interest
in \cite{Stojnicclupint19} but this very same analysis is of interest overthere as well. However, in \cite{Stojnicclupint19}, such an analysis isn't for the purpose of the discussion of the algorithm's computational complexity but is rather utilized to ensure that the algorithm misses a specific stationary point right at the beginning.

We recall that the CLuP's first step (iteration) amounts to determining $\x^{(1)}$ as
\begin{eqnarray}
\x^{(1)}=\frac{\x^{(1,s)}}{\|\x^{(1,s)}\|_2} \quad \mbox{with}\quad \x^{(1,s)}=\mbox{arg}\min_{\x} & & -(\x^{(0)})^T\x  \nonumber \\
\mbox{subject to} & & \|\y-A\x\|_2\leq r\nonumber \\
&& \x\in \left [-\frac{1}{\sqrt{n}},\frac{1}{\sqrt{n}}\right ]^n. \label{eq:avoidclup1}
\end{eqnarray}
The analysis of the above optimization problem is another simple exercise within the \bl{\textbf{Random Duality Theory}} (RDT). We will study way more involved problems than this one in one of our companion papers, and we may on occasion revisit this one again overthere in a bit more detail. Here though we will just sketch the RDT analysis relying on what we presented in \cite{Stojnicclupint19} and in a host of our earlier papers \cite{StojnicCSetam09,StojnicCSetamBlock09,StojnicISIT2010binary,StojnicDiscPercp13,StojnicUpper10,StojnicGenLasso10,StojnicGenSocp10,StojnicPrDepSocp10,StojnicRegRndDlt10,Stojnicbinary16fin,Stojnicbinary16asym}.
The above problem is of course very similar to the problems considered in \cite{Stojnicclupint19} and a majority of the discussion that applied overthere will be in place here again. As we are now not focused on all the key parameters that we considered in \cite{Stojnicclupint19} some steps can be done even faster. To that end we first quickly rewrite (\ref{eq:avoidclup1}) as
\begin{eqnarray}
\min_{\z} & & -(\x^{(0)})^T(\x_{sol}-\z)  \nonumber \\
\mbox{subject to} & & \|\sigma\v+A\z\|_2\leq r\nonumber \\
&& \z\in \left [0,2/\sqrt{n}\right ]^n, \label{eq:avoidclup1a1}
\end{eqnarray}
and further as
\begin{eqnarray}
\min_{\z} & & (\x^{(0)})^T\z  \nonumber \\
\mbox{subject to} & & \|\sigma\v+A\z\|_2\leq r\nonumber \\
&& \z\in \left [0,2/\sqrt{n}\right ]^n. \label{eq:avoidclup1a2}
\end{eqnarray}
Relying on the same concentration strategy that we introduced in \cite{StojnicCSetam09,StojnicCSetamBlock09,StojnicISIT2010binary,StojnicDiscPercp13,StojnicUpper10,StojnicGenLasso10,StojnicGenSocp10,StojnicPrDepSocp10,StojnicRegRndDlt10,Stojnicbinary16fin,Stojnicbinary16asym} and considered in \cite{Stojnicclupint19}, we here also set $\|\z\|_2^2=c_{1,z}$ and $(\x_{0})^T\z=s_1$. Then the following is the object of interest
\begin{eqnarray}
\xi_{p,1}(\alpha,\sigma,c_{1,z},s_1)=\lim_{n\rightarrow\infty}\frac{1}{\sqrt{n}}\mE \min_{\z} & & \|\sigma\v+A\z\|_2  \nonumber \\
\mbox{subject to} & & \|\z\|_2^2=c_{1,z}\nonumber \\
& & (\x^{(0)})^T\z=s_1 \nonumber \\
&& \z\in \left [0,2/\sqrt{n}\right ]^n. \label{eq:avoidclup1a3}
\end{eqnarray}
Now we are in position to mimic what we presented in early sections of \cite{Stojnicclupint19}. However, we do mention right here that below we will try to avoid as many of the unnecessary repetitive explanations as possible and instead focus on the key differences. As in \cite{Stojnicclupint19}, we again start by doing the RDT's first step which amounts to forming the deterministic Lagrange dual.

\vspace{.1in}
\noindent \xmyboxc{\bl{\emph{\textbf{1. First step -- \dgr{Forming the deterministic Lagrange dual} }}}}

\vspace{.1in}
Again following a large body of our earlier work we have (see, e.g. (\cite{StojnicCSetam09,StojnicISIT2010binary,StojnicDiscPercp13,StojnicGenLasso10,StojnicGenSocp10,StojnicPrDepSocp10,StojnicRegRndDlt10}))
\begin{eqnarray}
\min_{\z}\max_{\|\lambda\|_2=1,\gamma,\nu} & & \lambda^T\left ([A \quad \v]\begin{bmatrix}\z\\\sigma\end{bmatrix} \right )+\nu((\x^{(0)})^T\z-s_1)+\gamma (\|\x\|_2^2-c_{1,z}) \nonumber \\
\mbox{subject to} & & \z\in \left [0,2/\sqrt{n}\right ]^n. \label{eq:avoidclup8}
\end{eqnarray}
As in \cite{Stojnicclupint19}, given that we are interested in a statistical and large dimensional scenario, $\nu$ and $\gamma$ will concentrate and as scalars can be discretized so that the resulting optimization over these two quantities can be taken outside
\begin{eqnarray}
\max_{\gamma,\nu}\min_{\z}\max_{\|\lambda\|_2=1} & & \lambda^T\left ([A \quad \v]\begin{bmatrix}\z\\\sigma\end{bmatrix} \right )+\nu((\x^{(0)})^T\z-s_1)+\gamma (\|\x\|_2^2-c_{1,z}) \nonumber \\
\mbox{subject to} & & \z\in \left [0,2/\sqrt{n}\right ]^n. \label{eq:avoidclup9}
\end{eqnarray}

\vspace{.1in}
\noindent \xmyboxc{\bl{\emph{\textbf{2. Second step -- \dgr{Forming the Random dual} }}}}

\vspace{.1in}
Following further the principles of the analysis presented in \cite{Stojnicclupint19} we again introduce the so-called \bl{\textbf{random dual}}. Let $\bar{{\cal Z}}=\left [0,2/\sqrt{n}\right ]^n$. In this particular case the random dual is the following problem
\begin{equation}
\max_{\gamma,\nu}\min_{\z\in \bar{{\cal Z}}}\max_{\|\lambda\|_2=1} \lambda^T\g\sqrt{\|\z\|_2^2+\sigma^2}+\|\lambda\|_2(\h^T\z+h_0\sigma) +\nu((\x^{(0)})^T\z-s_1)+\gamma (\|\z\|_2^2-c_{1,z}),\\ \label{eq:avoidclup10}
\end{equation}
where we again have as earlier that the components of the newly introduced $m$ and $n$ dimensional vectors $\g$ and $\h$ are i.i.d. standard normals and $h_0$ is yet another standard normal independent of all other random variables. One here also observes that the minus sign in front of the second term typically present in some of the analysis in \cite{Stojnicclupint19} is not that important here and we remove it. Let $\xi_{RD}^{(1)}(\alpha,\sigma;c_{1,z},s_1,\gamma,\nu)$ be the following
\begin{equation}
\lim_{n\rightarrow\infty}\frac{1}{\sqrt{n}}\mE\min_{\z\in \bar{{\cal Z}}}\max_{\|\lambda\|_2=1} \lambda^T\g\sqrt{\|\z\|_2^2+\sigma^2}+\|\lambda\|_2(\h^T\z+h_0\sigma) +\nu((\x^{(0)})^T\z-s_1)+\gamma (\|\z\|_2^2-c_{1,z}). \label{eq:avoidclup10a}
\end{equation}

\vspace{.0in}
\noindent \xmyboxc{\bl{\emph{\textbf{3. Third step -- \dgr{Handling the Random dual} }}}}

\vspace{.1in}
Finally, in the third step we proceed with the analysis of the above \bl{\textbf{random dual}} following once again step by step the strategy outlined in \cite{StojnicCSetam09,StojnicISIT2010binary,StojnicDiscPercp13,StojnicGenLasso10,StojnicGenSocp10,StojnicPrDepSocp10,StojnicRegRndDlt10}.  The inner optimization over $\lambda$ is again very easy and one has
\begin{eqnarray}
\min_{c_1}\max_{\gamma,\nu}\min_{\z} & & \|\g\|_2\sqrt{c_{1,z}+\sigma^2}+(\h^T\z+h_0\sigma) +\nu((\x^{(0)})^T\z-s_1)+\gamma (\|\z\|_2^2-c_{1,z}) \nonumber \\
\mbox{subject to} & & \z\in \left [0,2/\sqrt{n}\right ]^n. \label{eq:avoidclup11}
\end{eqnarray}
Following again say \cite{StojnicDiscPercp13} we define
\begin{eqnarray}
f_{box,1}(\h;c_{1,z},s_1)=\max_{\gamma,\nu}\min_{\x} & & \h^T\z +\nu((\x^{(0)})^T\z-s_1)+\gamma (\|\z\|_2^2-c_{1,z}) \nonumber \\
\mbox{subject to} & & \z\in \left [0,2/\sqrt{n}\right ]^n.\label{eq:avoidclup12}
\end{eqnarray}
 With the addition of the $s_1$ constraint this is now literally identical to the box constrained problem considered in \cite{StojnicDiscPercp13} and we could immediately use the solution given there with a little bit of modification similar to the ones presented in \cite{Stojnicclupint19} to account for $s_1$ and $\nu$. Basically, instead of (110) from \cite{StojnicDiscPercp13} one now has
\begin{eqnarray}
f_{box,1}(\h;c_2,c_1)  =  \max_{\gamma,\nu} & & \frac{1}{\sqrt{n}}\left (\sum_{i=1}^{n}f_{box,1}^{(1)}(\h_i,\gamma,\nu)\right )-\nu s_1\sqrt{n}-\gamma c_{1,z}\sqrt{n},\label{eq:avoidclup13}
\end{eqnarray}
where
\begin{equation}
f_{box,1}^{(1)}(\h_i,\gamma,\nu)=\begin{cases}0, & \h_i+\nu\x^{(0)}_i\geq 0\\
-\frac{(\h_i+\nu\x^{(0)}_i)^2}{4\gamma}, & -4\gamma\leq \h_i+\nu\x^{(0)}_i\leq 0\\
2(\h_i+\nu\x^{(0)}_i)+4\gamma, & \h_i+\nu\x^{(0)}_i\leq -4\gamma,
\end{cases}\label{eq:avoidclup14}
\end{equation}
and $\gamma$ and $\nu$ are $\sqrt{n}$ scaled versions of $\gamma$ and $\nu$ from (\ref{eq:avoidclup12}) and $\x^{(0)}_i$ in (\ref{eq:avoidclup14}) is also $\sqrt{n}$ scaled (basically it is just the sign of the initial $\x^{(0)}_i$). Utilizing further the strategies outlined in \cite{Stojnicclupint19} one also has for the optimizing $\z_i$
\begin{equation}
\z_i=\frac{1}{\sqrt{n}}\min \left (\max\left (0,-\left (\frac{\h+\nu\x^{(0)}_i}{2\gamma}\right )\right ),2\right ).\label{eq:avoidclup14a}
\end{equation}
Assuming that the initial $\x^{(0)}$ has $\rho n$ components equal to $\frac{1}{\sqrt{n}}$ and $(1-\rho) n$ components equal to $-\frac{1}{\sqrt{n}}$ and after solving the integrals one has
\begin{equation}
\mE f_{box,1}^{(1)}(\h_i,\gamma,\nu)=\rho I_{1,1}(\gamma,\nu)+(1-\rho) I_{1,1}(\gamma,-\nu)+\rho I_{2,1}(\gamma,\nu)+(1-\rho) I_{2,1}(\gamma,-\nu),\label{eq:avoidclup15}
\end{equation}
where
\begin{eqnarray}
I_{1,1}(\gamma,\nu) &  = & -(exp(-0.5(4 \gamma + \nu)^2) (\nu - 4 \gamma) + \sqrt{\pi/2} (\nu^2 + 1) \erf(2\sqrt{2}\gamma + 1/\sqrt{2}\nu) \nonumber \\
& & - \sqrt{\pi/2} (\nu^2 + 1) \erf(\nu/\sqrt{2}) - exp(-0.5 \nu^2) \nu)/(4 \sqrt{2 \pi}\gamma) \nonumber \\
I_{2,1}(\gamma,\nu)  &  = & (4\gamma+2\nu).5\erfc((4\gamma+\nu)/\sqrt{2})-2exp(-1/2(4\gamma+\nu)^2)/\sqrt{2\pi}.
\label{eq:avoidclup16}
\end{eqnarray}
Finally a combination of (\ref{eq:avoidclup10})-(\ref{eq:avoidclup16}) gives
\begin{equation}
\xi_{RD}^{(1)}(\alpha,\sigma;c_{1,z},s_1,\gamma,\nu)=\sqrt{\alpha}\sqrt{c_{1,z}+\sigma^2}+\mE f_{box,1}^{(1)}(\h_i,\gamma,\nu)-\nu s_1-\gamma c_{1,z}. \label{eq:avoidclup17}
\end{equation}
The following theorem summarizes what we presented above.
\begin{theorem}(CLuP -- RDT estimate -- first iteration)
Let $\xi_{p,1}(\alpha,\sigma,c_{1,z},s_1)$ and $\xi_{RD}^{(1)}(\alpha,\sigma;c_2,c_1,\gamma,\nu)$ be as in (\ref{eq:avoidclup1a3}) and (\ref{eq:avoidclup17}), respectively. Then \begin{equation}
\xi_{p,1}(\alpha,\sigma,c_{1,z},s_1)= \max_{\gamma,\nu}\xi_{RD}^{(1)}(\alpha,\sigma;c_{1,z},s_1,\gamma,\nu).\label{eq:thmavoidcluprd1}
\end{equation}
Consequently,
\begin{equation}
\min_{c_{1,z}}\xi_{p,1}(\alpha,\sigma,c_{1,z},s_1)= \min_{c_{1,z}}\max_{\gamma,\nu}\xi_{RD}^{(1)}(\alpha,\sigma;c_{1,z},s_1,\gamma,\nu).
\label{eq:thmcluprd2}
\end{equation}\label{thm:avoidcluprd1}
\end{theorem}
\begin{proof}
Follows from the above derivation and the general RDT concepts presented in \cite{StojnicCSetam09,StojnicISIT2010binary,StojnicDiscPercp13,StojnicGenLasso10,StojnicGenSocp10,StojnicPrDepSocp10,StojnicRegRndDlt10} and the fact that the \bl{\textbf{strong random duality}} trivially holds.
\end{proof}
Given that the strong random duality is in place one can continue further and obtain the exact estimates for all other relevant quantities. We formalize that below.

\subsection{CLuP -- summary of the first iteration performance analysis}
\label{sec:clupfirstitsummary}

Given that the above presentation may have gone a bit too much into a mathematical analysis, we below present a few keys steps one needs to perform to actually calculate pretty much all quantities of interest. Of course, we do emphasize that all of that is possible precisely because of the above analysis.

\vspace{.1in}
\noindent \xmyboxc{\bl{\emph{\textbf{Summarized formalism to handle the CLuP's first iteration }}}}

First consider the following optimization problem
\begin{eqnarray}
\{\hat{\nu}^{(1)},\hat{\gamma}^{(1)},\hat{c}_{1,z}^{(1)},\hat{s}_1^{(1)}\}=\mbox{arg} \min_{s_1} & & s_1 \nonumber \\
\mbox{subject to} & & \min_{0\leq c_{1,z}\leq 4}\max_{\gamma,\nu}\xi_{RD}^{(1)}(\alpha,\sigma;c_{1,z},s_1,\gamma,\nu)= r.\label{eq:avoidclup17a}
\end{eqnarray}Let
\begin{eqnarray}
s_{x,1}(\gamma,\nu) & = & -\nu/2/\gamma(.5\erfc(\nu/\sqrt{2})-.5\erfc((\nu+4\gamma)/\sqrt{2}))\nonumber \\
& & +1/2/\gamma/\sqrt{2\pi}(exp(-\nu^2/2)-exp(-(4\gamma+\nu)^2/2))\nonumber \\
s_{xsq,1}(\gamma,\nu) & = & -I_{1,1}(\gamma,\nu)/\gamma\nonumber \\
s_{x,2}(\gamma,\nu) & = & 2(.5\erfc((4\gamma+\nu)/\sqrt{2}))\nonumber \\
s_{xsq,2}(\gamma,\nu) & = & 2s_{x,2}.\label{eq:avoidclup18}
\end{eqnarray}
Utilizing (\ref{eq:avoidclup14a}) we have
\begin{eqnarray}
\sqrt{n}\mE\z_i & = & \rho s_{x,1}(\hat{\gamma}^{(1)},\hat{\nu}^{(1)})+(1-\rho)s_{x,1}(\hat{\gamma}^{(1)},-\hat{\nu}^{(1)})+\rho s_{x,2}(\hat{\gamma}^{(1)},\hat{\nu}^{(1)})+(1-\rho)s_{x,2}(\hat{\gamma}^{(1)},-\hat{\nu}^{(1)})\nonumber \\
n\mE\z_i^2 & = & \rho s_{xsq,1}(\hat{\gamma}^{(1)},\hat{\nu}^{(1)})+(1-\rho)s_{xsq,1}(\hat{\gamma}^{(1)},-\hat{\nu}^{(1)})+\rho s_{xsq,2}(\hat{\gamma}^{(1)},\hat{\nu}^{(1)})+(1-\rho)s_{xsq,2}(\hat{\gamma}^{(1)},-\hat{\nu}^{(1)}).\nonumber \\\label{eq:avoidclup19}
\end{eqnarray}
Given that $\x_i=\x_{sol}-\z_i$ we finally also have
\begin{eqnarray}
\sqrt{n}\mE\x_i & = & 1-(\rho s_{x,1}(\hat{\gamma}^{(1)},\hat{\nu}^{(1)})+(1-\rho)s_{x,1}(\hat{\gamma}^{(1)},-\hat{\nu}^{(1)})+\rho s_{x,2}(\hat{\gamma}^{(1)},\hat{\nu}^{(1)})+(1-\rho)s_{x,2}(\hat{\gamma}^{(1)},-\hat{\nu}^{(1)}))\nonumber \\
n\mE\x_i^2 & = & \rho s_{xsq,1}(\hat{\gamma}^{(1)},\hat{\nu}^{(1)})+(1-\rho)s_{xsq,1}(\hat{\gamma}^{(1)},-\hat{\nu}^{(1)})+\rho s_{xsq,2}(\hat{\gamma}^{(1)},\hat{\nu}^{(1)})+(1-\rho)s_{xsq,2}(\hat{\gamma}^{(1)},-\hat{\nu}^{(1)})\nonumber \\
&&+2\mE\x_i-1,\label{eq:avoidclup20}
\end{eqnarray}
and
\begin{eqnarray}
\mE((\x_{sol})^T\x) \hspace{-.08in}& = & \hspace{-.08in}1-(\rho s_{x,1}(\hat{\gamma}^{(1)},\hat{\nu}^{(1)})+(1-\rho)s_{x,1}(\hat{\gamma}^{(1)},-\hat{\nu}^{(1)})+\rho s_{x,2}(\hat{\gamma}^{(1)},\hat{\nu}^{(1)})+(1-\rho)s_{x,2}(\hat{\gamma}^{(1)},-\hat{\nu}^{(1)}))\nonumber \\
\mE\|\x\|_2^2 \hspace{-.08in} & = & \hspace{-.08in} \rho s_{xsq,1}(\hat{\gamma}^{(1)},\hat{\nu}^{(1)})+(1-\rho)s_{xsq,1}(\hat{\gamma}^{(1)},-\hat{\nu}^{(1)})+\rho s_{xsq,2}(\hat{\gamma}^{(1)},\hat{\nu}^{(1)})+(1-\rho)s_{xsq,2}(\hat{\gamma}^{(1)},-\hat{\nu}^{(1)})\nonumber \\
&&+2\mE((\x_{sol})^T\x)-1.\label{eq:avoidclup21}
\end{eqnarray}
Of course, given that the strong random duality is in full power the above quantities are actually the concentrating values and the concentration is exponential in $n$. Finally, utilizing (\ref{eq:avoidclup14a}) one can also get the estimate for the probability of error
\begin{equation}\label{eq:avoidclup21a}
  p_{err}^{(1)}=1-P\left (\z_i\leq \frac{1}{\sqrt{n}}\right )=1-\left (\rho\left (\frac{1}{2}\erfc\left ( \frac{-2\hat{\gamma}^{(1)}-\hat{\nu}^{(1)}}{\sqrt{2}}\right )\right )+(1-\rho)\left ( \frac{1}{2}\erfc\left ( \frac{-2\hat{\gamma}^{(1)}+\hat{\nu}^{(1)}}{\sqrt{2}}\right )\right )\right ).
\end{equation}

\subsection{Numerical results -- first iteration}
\label{sec:clupfirstnum}

In this section we present a set of numerical results that relate to the above analysis. Numerical analysis is needed for both, the theoretical values discussed above and for their simulated counterparts. We start, with the theoretical predictions and the numerical evaluations of the critical parameters discussed in the analysis presented above. In Table \ref{tab:tabavoidnum1} we show the theoretical values for various system parameters obtained based on Theorem \ref{thm:avoidcluprd1} for two different values of SNR, $1/\sigma^2=10$[db] and $1/\sigma^2=13$[db].
\begin{table}[h]
\caption{\textbf{Theoretical} values for various system parameters obtained utilizing Theorem \ref{thm:avoidcluprd1}}\vspace{.1in}
\hspace{-0in}\centering
\small{
\begin{tabular}{||c||c|c||c|c||c|c|c|c||}\hline\hline
$1/\sigma^2 $[db]  & $\hat{\nu}^{(1)}$ & $\hat{\gamma}^{(1)}$ & $\hat{c}_{1,z}^{(1)}$ & $\hat{s}_1^{(1)}$ &   $\xi_{RD}^{(1)} $ & $p_{err}^{(1)} $  & $\|\x^{(1,s)}\|_2^2$ &
$(\x_{sol})^T\x^{(1,s)}$ \\ \hline\hline
$10  $ & $\mathbf{0.5075}  $ & $\mathbf{0.6816}  $ & $\mathbf{0.3306}  $ & $\mathbf{-0.1844}  $ & $\mathbf{0.2252}  $ & $\mathbf{0.1134}  $ & $\mathbf{0.6749}  $ & $\mathbf{0.6722}  $ \\ \hline
$13  $ & $\mathbf{0.4953}  $ & $\mathbf{0.9420}  $ & $\mathbf{0.1753}  $ & $\mathbf{-0.1314}  $ & $\mathbf{0.1594}  $ & $\mathbf{0.0456}  $ & $\mathbf{0.7009}  $ & $\mathbf{0.7628}  $ \\ \hline\hline
\end{tabular}}
\label{tab:tabavoidnum1}
\end{table}

\subsubsection{Simulations -- first iteration}
\label{sec:clupfirstsim}

To give a bit of a feeling as to what kind of precision level the \bl{\textbf{random duality theory}} typically achieves even for moderate problem dimensions we in Table \ref{tab:tabavoidnum2} show the corresponding simulated values. The simulated values are obtained for $\alpha=0.8$, $n=400$, and $r_{sc}=1.3$. Given that $r_{plt}=.1226$ this implies that $\xi_{RD}^{(1)}=r=r_{sc}r_{plt}=.1594$.
\begin{table}[h]
\caption{\textbf{Theoretical}/\bl{\textbf{simulated}} values for various system parameters obtained based on Theorem \ref{thm:avoidcluprd1}}\vspace{.1in}
\hspace{-0in}\centering
\small{
\begin{tabular}{||c||c||c|c|c|c||}\hline\hline
$1/\sigma^2 $[db]  &  $\hat{s}_1^{(1)}$ &   $\xi_{RD}^{(1)} $ & $p_{err}^{(1)} $  & $\|\x^{(1,s)}\|_2^2$ &
$(\x_{sol})^T\x^{(1,s)}$ \\ \hline\hline
$10  $ & $\mathbf{-0.1844  }/\bl{\mathbf{0.1845  }}$ & $\mathbf{-0.2252  }/\bl{\mathbf{0.2252  }}$ & $\mathbf{0.1134  }/\bl{\mathbf{0.1133  }}$ & $\mathbf{0.6749  }/\bl{\mathbf{0.6769  }}$ & $\mathbf{0.6722  }/\bl{\mathbf{0.6719  }}$ \\ \hline
$13  $ & $\mathbf{-0.1314  }/\bl{\mathbf{0.1316  }}$ & $\mathbf{-0.1594  }/\bl{\mathbf{0.1594  }}$ & $\mathbf{0.0456  }/\bl{\mathbf{0.0483  }}$ & $\mathbf{0.7009  }/\bl{\mathbf{0.7005  }}$ & $\mathbf{0.7628  }/\bl{\mathbf{0.7596  }}$ \\ \hline\hline
\end{tabular}}
\label{tab:tabavoidnum2}
\end{table}
The results given in Tables \ref{tab:tabavoidnum1} and \ref{tab:tabavoidnum2} can be utilized (as discussed in \cite{Stojnicclupint19}) to ensure that one circumvents one of the stationary points right at the beginning after the first iteration. Looking at the tables one observes $c_2=0.6749$ (or if one views the simulated value $c_2=0.6769$) which is substantially above $0.46075$ given in \cite{Stojnicclupint19}. As concluded in \cite{Stojnicclupint19}, since the CLuP's objective never decreases one is then indeed reassured that the lower stationary point will be circumvented. Of course, as mentioned above and since it is basically trivial within the concepts of the \bl{\textbf{random duality theory}} we will not repeatedly stress throughout the paper but do mention here that whenever we determine certain expected value of a quantity associated with a vector as a whole that will be the concentrating point and the underlying concentration is rather overwhelming (exponential in $n$).

We also in Table \ref{tab:tabavoidnum3} show how the results change as the dimension $n$ changes. As we have mentioned above, already for $n=400$ the results are fairly close to the theoretical predictions. For $n=1600$ they are almost identical across all key parameters. Also, to make notation a bit easier in the above derivation we skip utilizing indices to emphasize that we are discussing the first iteration. In tables we added them as superscripts with a rather obvious meaning. Throughout the paper we will try to maintain the same approach and when not necessary we may actually skip adding iteration or other indices.
\begin{table}[h]
\caption{\bl{\textbf{Simulated}} values for $p_{err}^{(1)}  $, $\hat{s}_{1}^{(1)}$, $\|\x^{(1,s)}\|_2^2$, and $(\x_{sol})^T\x^{(1,s)}$; $\alpha=0.8$, $r_{sc}=1.3$, $n=\{100,200,400,800,1600\}$}\vspace{.1in}
\hspace{-0in}\centering
\footnotesize{
\begin{tabular}{||c||c||c|c|c|c||}\hline\hline
$ n $ & $\#$ of reps.   & $p_{err}^{(1)}  $ & $\hat{s}_{1}^{(1)}$ & $\|\x^{(1,s)}\|_2^2  $ & $(\x_{sol})^T\x^{(1,s)}  $  \\ \hline\hline
$100  $ & $165  $ & $\bl{\mathbf{0.0485}}  $ & $\bl{\mathbf{0.1230}}  $ & $\bl{\mathbf{0.7097}}  $ & $\bl{\mathbf{0.7640}}  $ \\ \hline
$200  $ & $954  $ & $\bl{\mathbf{0.0495}}  $ & $\bl{\mathbf{0.1314}}  $ & $\bl{\mathbf{0.7046}}  $ & $\bl{\mathbf{0.7601}}  $ \\ \hline
$400  $ & $600  $ & $\bl{\mathbf{0.0483}}  $ & $\bl{\mathbf{0.1316}}  $ & $\bl{\mathbf{0.7005}}  $ & $\bl{\mathbf{0.7596}}  $ \\ \hline
$800  $ & $465  $ & $\bl{\mathbf{0.0447}}  $ & $\bl{\mathbf{0.1294}}  $ & $\bl{\mathbf{0.7046}}  $ & $\bl{\mathbf{0.7658}}  $ \\ \hline
$1600  $ & $170  $ & $\bl{\mathbf{0.0454}}  $ & $\bl{\mathbf{0.1308}}  $ & $\bl{\mathbf{0.7028}}  $ & $\bl{\mathbf{0.7642}}  $ \\ \hline\hline
$\infty$ -- \textbf{theory} & $-$ & $\mathbf{0.0456}$  & $\mathbf{0.1314}$ & $\mathbf{0.7009}$ & $\mathbf{0.7628}$ \\ \hline\hline
\end{tabular}}
\label{tab:tabavoidnum3}
\end{table}

\section{CLuP -- second iteration performance analysis}
\label{sec:clupsecondit}

The above analysis is of course related to the first iteration of the algorithm. In this section we discuss the second iteration. This is the major, key step in the discussion of the complexity of the entire algorithm and all other steps are basically just a simple generalization of the foundational concepts needed for this second step. Some of the considerations that we present below will be similar to those presented when we discussed the first iteration. On the other hand some of them will be very different. To ensure the easiness of the exposition we will try to rely on what we presented above as much as possible. However, we will also try to avoid repetitive reexplaining of the already introduced concepts and, as usual, will try to put the emphasis on the key differences.

We start by recalling that the CLuP's second iteration amounts to determining $\x^{(2)}$ as
\begin{eqnarray}
\x^{(2)}=\frac{\x^{(2,s)}}{\|\x^{(2,s)}\|_2} \quad \mbox{with}\quad \x^{(2,s)}=\mbox{arg}\min_{\x} & & -(\x^{(1)})^T\x  \nonumber \\
\mbox{subject to} & & \|\y-A\x\|_2\leq r\nonumber \\
&& \x\in \left [-\frac{1}{\sqrt{n}},\frac{1}{\sqrt{n}}\right ]^n. \label{eq:avoidsecclup1}
\end{eqnarray}
We then closely follow what was done earlier and first rewrite (\ref{eq:avoidsecclup1}) as
\begin{eqnarray}
\min_{\z} & & -(\x^{(1)})^T(\x_{sol}-\z)  \nonumber \\
\mbox{subject to} & & \|\sigma\v+A\z\|_2\leq r\nonumber \\
&& \z\in \left [0,2/\sqrt{n}\right ]^n, \label{eq:avoidsecclup1a1}
\end{eqnarray}
and further as
\begin{eqnarray}
\min_{\z} & & (\x^{(1)})^T\z  \nonumber \\
\mbox{subject to} & & \|\sigma\v+A\z\|_2\leq r\nonumber \\
&& \z\in \left [0,2/\sqrt{n}\right ]^n. \label{eq:avoidsecclup1a2}
\end{eqnarray}
As in the previous section, we will again rely on the same concentration strategy  introduced in \cite{StojnicCSetam09,StojnicCSetamBlock09,StojnicISIT2010binary,StojnicDiscPercp13,StojnicUpper10,StojnicGenLasso10,StojnicGenSocp10,StojnicPrDepSocp10,StojnicRegRndDlt10,Stojnicbinary16fin,Stojnicbinary16asym} and considered in \cite{Stojnicclupint19}, and set $\|\z\|_2^2=c_{2,z}$ and $(\x^{(1,s)})^T\z=s_2$. Analogously to (\ref{eq:avoidclup1a3}) one can then view the following optimization problem as the object of interest
\begin{eqnarray}
\xi_{p,2}(\alpha,\sigma,c_{2,z},s_2)=\lim_{n\rightarrow\infty}\frac{1}{\sqrt{n}}\mE\min_{\z} & & \|\sigma\v+A\z\|_2  \nonumber \\
\mbox{subject to} & & \|\z\|_2^2=c_{2,z}\nonumber \\
& & (\x^{(1,s)})^T\z=s_2 \nonumber \\
&& \z\in \left [0,2/\sqrt{n}\right ]^n. \label{eq:avoidsecclup1a3}
\end{eqnarray}
This problem of course structurally completely matches the one in (\ref{eq:avoidclup1a3}). So, conceivably, one can then proceed with an analysis similar to the one presented in the previous section right after (\ref{eq:avoidclup1a3}). That is as an approximation possible but is likely to lead to inaccurate estimates already on the level of the second iteration. Those potential inaccuracies would have a chance to be even more pronounced after propagations through later iterations (plus one would have to likely face structurally similar problems in later iterations as well and if similar approximations are to be made again they might introduce another set of inaccuracies on their own that could also become more pronounced after going through the iterations that would follow). So, where is actually the core of the problem? Namely, while the problems in (\ref{eq:avoidclup1a3}) and (\ref{eq:avoidsecclup1a3}) do look almost identical, they also have one key difference. Instead of $\x^{(0)}$ that one has in (\ref{eq:avoidclup1a3}), in (\ref{eq:avoidsecclup1a3}) one has $\x^{(1)}$. If $\x^{(1)}$ were a constant or randomly generated the approach from the previous section could be used; it is just that it would have to be slightly adjusted. However, the problem one faces here is much bigger. It is not just that $\x^{(1)}$ is different because of the way how it is generated, its randomness actually depends on the problem structure and can not in principle be separated from it as it was in the previous section for $\x^{(0)}$.

We will eventually provide a way to estimate $\xi_{p,2}(\alpha,\sigma,c_{2,z},s_2)$. However, a lot of work will be needed before reaching the point to be able to do that. So, instead of jumping directly to (\ref{eq:avoidsecclup1a3}) one can rewrite (\ref{eq:avoidsecclup1a2}) in the following way
\begin{eqnarray}
\min_{\z,\z^{(1)}} & & (\x^{(1)})^T\z  \nonumber \\
\mbox{subject to} & & \|\sigma\v+A\z\|_2\leq r\nonumber \\
&& \z\in \left [0,2/\sqrt{n}\right ]^n \nonumber \\
& & \|\sigma\v+A\z^{(1)}\|_2\leq r\nonumber \\
 & & \|\z^{(1)}\|_2^2=\hat{c}_{1,z}^{(1)}\nonumber \\
& & (\x^{(0)})^T\z^{(1)}=\hat{s}_1^{(1)} \nonumber \\
&& \z^{(1)}\in \left [0,2/\sqrt{n}\right ]^n, \label{eq:avoidsecclup1a2a}
\end{eqnarray}
where the bottom portion of the constraints is pretty much artificially added for the completeness. From the analysis in the previous section it is clear that there is not really much of freedom to optimize over $\z^{(1)}$. One can now proceed with the standard RDT steps. As usual, the first one amounts to forming the deterministic Lagrange dual.

\vspace{.1in}
\noindent \xmyboxc{\bl{\emph{\textbf{1. First step -- \dgr{Forming the deterministic Lagrange dual} }}}}

\vspace{.1in}
Also as usual, we once again follow a large body of our earlier work and obtain (see, e.g. \cite{StojnicCSetam09,StojnicISIT2010binary,StojnicDiscPercp13,StojnicGenLasso10,StojnicGenSocp10,StojnicPrDepSocp10,StojnicRegRndDlt10})
\begin{eqnarray}
\min_{\z,\z^{(1)}}\max_{\gamma,\gamma_0\geq 0} & & (\x^{(1)})^T\z+\gamma(\|\sigma\v+A\z\|_2- r)+\gamma_{0}(\|\sigma\v+A\z^{(1)}\|_2- r) \nonumber \\
\mbox{subject to} && \z\in \left [0,2/\sqrt{n}\right ]^n \nonumber \\
 & & \|\z^{(1)}\|_2^2=\hat{c}_{1,z}^{(1)}\nonumber \\
& & (\x^{(0)})^T\z^{(1)}=\hat{s}_1^{(1)} \nonumber \\
&& \z^{(1)}\in \left [0,2/\sqrt{n}\right ]^n, \label{eq:avoidsecclup1a2a1}
\end{eqnarray}
After a couple of cosmetic changes one also has
\begin{eqnarray}
\min_{\z,\z^{(1)}} \max_{\gamma,\gamma_0\geq 0}\max_{\|\lambda\|_2=1,\|\lambda_0\|_2=1}& & (\x^{(1)})^T\z+\gamma\z_{sc}\lambda^T\begin{bmatrix}A &\v\end{bmatrix}\begin{bmatrix}\z\\\sigma\end{bmatrix}/\z_{sc}- \gamma r+\gamma_0\z_{sc}^{(1)}\lambda_0^T\begin{bmatrix}A &\v\end{bmatrix}\begin{bmatrix}\z^{(1)}\\\sigma\end{bmatrix}/\z_{sc}^{(1)}- \gamma_0 r \nonumber \\
\mbox{subject to} && \z\in \left [0,2/\sqrt{n}\right ]^n, \z_{sc}=\sqrt{\|\z\|_2^2+\sigma^2} \nonumber \\
 & & \|\z^{(1)}\|_2^2=\hat{c}_{1,z}^{(1)}\nonumber \\
& & (\x^{(0)})^T\z^{(1)}=\hat{s}_1^{(1)} \nonumber \\
&& \z^{(1)}\in \left [0,2/\sqrt{n}\right ]^n, \z_{sc}^{(1)}=\sqrt{\|\z^{(1)}\|_2^2+\sigma^2}.\label{eq:avoidsecclup1a2a2}
\end{eqnarray}
Relying on the concentration of $\gamma$ and $\gamma_0$ one further has
\begin{eqnarray}
\max_{\gamma,\gamma_0\geq 0}\min_{\z,\z^{(1)}} \max_{\|\lambda\|_2=1,\|\lambda_0\|_2=1}& & (\x^{(1)})^T\z+\gamma\z_{sc}\lambda^T\begin{bmatrix}A &\v\end{bmatrix}\begin{bmatrix}\z\\\sigma\end{bmatrix}/\z_{sc}- \gamma r+\gamma_0\z_{sc}^{(1)}\lambda_0^T\begin{bmatrix}A &\v\end{bmatrix}\begin{bmatrix}\z^{(1)}\\\sigma\end{bmatrix}/\z_{sc}^{(1)}- \gamma_0 r \nonumber \\
\mbox{subject to} && \z\in \left [0,2/\sqrt{n}\right ]^n, \z_{sc}=\sqrt{\|\z\|_2^2+\sigma^2} \nonumber \\
 & & \|\z^{(1)}\|_2^2=\hat{c}_{1,z}^{(1)}\nonumber \\
& & (\x^{(0)})^T\z^{(1)}=\hat{s}_1^{(1)} \nonumber \\
&& \z^{(1)}\in \left [0,2/\sqrt{n}\right ]^n, \z_{sc}^{(1)}=\sqrt{\|\z^{(1)}\|_2^2+\sigma^2}.\label{eq:avoidsecclup1a2a3}
\end{eqnarray}

\vspace{.1in}
\noindent \xmyboxc{\bl{\emph{\textbf{2. Second step -- \dgr{Forming the Random dual} }}}}

\vspace{.1in}
Continuing to follow further the analysis presented in \cite{Stojnicclupint19} and the principles introduced earlier in e.g. \cite{StojnicCSetam09,StojnicISIT2010binary,StojnicDiscPercp13,StojnicGenLasso10,StojnicGenSocp10,StojnicPrDepSocp10,StojnicRegRndDlt10})
we can again introduce the so-called \bl{\textbf{random dual}}. However, things are much more subtle this time. Let $\bar{{\cal Z}}=\left [0,2/\sqrt{n}\right ]^n$. One considers then the following object
\begin{equation}\label{eq:avoidsecclup1a2a4}
  f_{RD}= (\x^{(1)})^T\z+ \gamma \z_{sc} f_{RD,2}- \gamma r+\gamma_0 \z_{sc}^{(1)} f_{RD,1}- \gamma_0 r,\nonumber \\
\end{equation}
where
\begin{eqnarray}
f_{RD,2} & = & \lambda^T(q^{(1)}\g+\sqrt{1-(q^{(1)})^2}\g^{(1)})+((p^{(1)}\h+\sqrt{1-(p^{(1)})^2}\h^{(1)})\z/\z_{sc}\nonumber\\
& & +(p^{(1)}h_0+\sqrt{1-(p^{(1)})^2}h_0^{(1)})\sigma/\z_{sc}), \label{eq:avoidsecclup1a2a5}
\end{eqnarray}
and
\begin{eqnarray}
f_{RD,1} & = & \lambda_0^T\g+\h^T\z^{(1)}/\z_{sc}^{(1)} +h_0\sigma/\z_{sc}^{(1)}, \label{eq:avoidsecclup1a2a6}
\end{eqnarray}
and as usual, the components of all $\g$, $\h$, $\g^{(1)}$, and $\h^{(1)}$ and $h_0^{(1)}$ are i.i.d. standard normals. It is not that hard to guess that the corresponding \bl{\textbf{random dual}} is then
\begin{eqnarray}
\min_{q^{(1)}}\max_{p^{(1)}}\max_{\gamma,\gamma_0\geq 0}\min_{\z,\z^{(1)}} \max_{\|\lambda\|_2=1,\|\lambda_0\|_2=1}& & f_{RD} \nonumber \\
\mbox{subject to} && \z\in \left [0,2/\sqrt{n}\right ]^n, \z_{sc}=\sqrt{\|\z\|_2^2+\sigma^2} \nonumber \\
 & & \|\z^{(1)}\|_2^2=\hat{c}_{1,z}^{(1)}\nonumber \\
& & (\x^{(0)})^T\z^{(1)}=\hat{s}_1^{(1)} \nonumber \\
&& \z^{(1)}\in \left [0,2/\sqrt{n}\right ]^n, \z_{sc}^{(1)}=\sqrt{\|\z^{(1)}\|_2^2+\sigma^2}\nonumber \\
& &  \begin{bmatrix}\z\\\sigma\end{bmatrix}^T\begin{bmatrix}\z^{(1)}\\\sigma\end{bmatrix}/\z_{sc}/\z_{sc}^{(1)}= q^{(1)}\nonumber \\
& & \lambda^T\lambda_0=p^{(1)}.\label{eq:avoidsecclup1a2a7}
\end{eqnarray}
There are a couple of things we should now add. If one looks solely at $f_{RD,1}$ that does seem perfectly fine on its own. It in fact is exactly the portion of the random dual that corresponds to the $\z^{(1)}$ portion of the objective in (\ref{eq:avoidsecclup1a2a3}) (an easy comparison with what was done in the previous section would quickly confirm that). Analogously, one then expects that similar object should be the portion of the random dual that corresponds to the $\z$ portion of the objective in (\ref{eq:avoidsecclup1a2a3}). That is of course exactly $f_{RD,2}$. Everything would be rather smooth if there were no $p^{(1)}$, $q^{(1)}$. The question is of course where these come from. That would be obviously very hard to understand right now and even way harder to explain without going into heavy machinery of the underlying fundamentals of the random duality theory itself. As such a discussion would overtake over the main topic of this paper, which is the complexity analysis of the CLuP algorithm, we leave it for a separate paper where we will revisit some of the random duality theory fundamentals. Here though, we just briefly mention that the meaning of the $p^{(1)}$, $q^{(1)}$ parameters is the following: $q^{(1)}$ is roughly speaking the presumed concentrating value of the the so-called optimal achieving $\begin{bmatrix}\z\\\sigma\end{bmatrix}$-cross-overlap, i.e.
\begin{equation}\label{eq:avoidsecclup1a2a8}
  q^{(1)}\approx\begin{bmatrix}\z\\\sigma\end{bmatrix}^T\begin{bmatrix}\z^{(1)}\\\sigma\end{bmatrix}/\z_{sc}/\z_{sc}^{(1)},
\end{equation}
and $p^{(1)}$ is roughly speaking the presumed concentrating value of the the so-called optimal achieving $\lambda$-cross-overlap, i.e.
\begin{equation}\label{eq:avoidsecclup1a2a9}
  p^{(1)}\approx\lambda^T\lambda_0.
\end{equation}
Of course, as one may guess, things are actually way more complicated since the above mentioned concentrating values are not just over the standard randomness but also over certain the so-called Gibbsian measures randomness as well (both $\z$'s and both $\lambda$'s above in such situations are running over all allowed $\z$'s and $\lambda$'s). Those types of measures and their randomness appear within the random duality theory as some of the most crucial objects and are way harder mathematical concepts than the regular Gaussian ones discussed above. As mentioned above, to avoid being sidetracked with all these mathematical complications we leave more detailed discussions in these directions for separate companion papers.

One can then proceed with handling (\ref{eq:avoidsecclup1a2a7}). That is in principle simple if one follows what we did in the previous section and earlier in \cite{Stojnicclupint19} and ultimately in e.g. \cite{StojnicCSetam09,StojnicISIT2010binary,StojnicDiscPercp13,StojnicGenLasso10,StojnicGenSocp10,StojnicPrDepSocp10,StojnicRegRndDlt10}). However, just looking at the problem one quickly observes that there are quite a few variables to optimize over. So, instead of directly working with (\ref{eq:avoidsecclup1a2a7}) we will in a way emulate what we did in \cite{Stojnicclupint19} and try to work through a few shortcuts. In equation (8) of \cite{Stojnicclupint19} we essentially had the type of problem that we have here in (\ref{eq:avoidsecclup1a2a7}). Instead of solving it directly we in \cite{Stojnicclupint19} created a shortcut mechanism starting from equation (9) and continuing further through Section 2.1.1 (later on, in Section 3.2.2, we revisited the problem from equation (8) and solved it directly as well). Here we will try to mimic the idea from equation (9) and Section 2.1.1 of \cite{Stojnicclupint19}. We do mention though that while in \cite{Stojnicclupint19} this turned out to be the best mechanism here one may be able to find even better ones. However, as we will see later on even the mechanism that we present below works very well. There will be the two main steps in the analysis process.

\vspace{.1in}
\noindent \xmyboxc{\bl{\emph{\textbf{3. Third step -- \dgr{Rehandling the Random dual of the first iteration} }}}}

To start things off we will first separately handle the following problem
\begin{eqnarray}
\min_{\z^{(1)}} \max_{\|\lambda_0\|_2=1}& & \z_{sc}^{(1)}f_{RD,1} \nonumber \\
\mbox{subject to}  & & \|\z^{(1)}\|_2^2=\hat{c}_{1,z}^{(1)}\nonumber \\
& & (\x^{(0)})^T\z^{(1)}=\hat{s}_1^{(1)} \nonumber \\
&& \z^{(1)}\in \left [0,2/\sqrt{n}\right ]^n, \z_{sc}^{(1)}=\sqrt{\|\z^{(1)}\|_2^2+\sigma^2},\label{eq:avoidsecclup1a2a10}
\end{eqnarray}
and utilize the obtained $\z_{sc}^{(1)}$. Not only does this emulate what we did in \cite{Stojnicclupint19}, it in a way also emulates the natural flow of the CLuP algorithm. It is now beyond trivial to recognize that the solution of (\ref{eq:avoidsecclup1a2a10}) is exactly what we obtained in the previous section. To be more precise, from (\ref{eq:avoidclup14a}) one actually has
\begin{equation}
\z_i^{(1)}=\frac{1}{\sqrt{n}}\min \left (\max\left (0,-\left (\frac{\h_i+\nu\x^{(0)}_i}{2\gamma}\right )\right ),2\right ),\label{eq:avoidsecclup1a2a11}
\end{equation}
and
\begin{equation}
\x_i^{(1,s)}=1-\z_i^{(1)}=\frac{1}{\sqrt{n}}\left (1-\min \left (\max\left (0,-\left (\frac{\h+\hat{\nu}^{(1)}\x^{(0)}_i}{2\hat{\gamma}^{(1)}}\right )\right ),2\right )\right ),\label{eq:avoidsecclup1a2a12}
\end{equation}
where from this point on to lighten writing we actually instead of $\x_{sol}$ assume its value of all ones scaled by $\sqrt{n}$.

\vspace{.2in}
\noindent \xmyboxc{\bl{\emph{\textbf{4. Fourth step -- \dgr{Handling the real Random dual of the second iteration} }}}}

Finally we are in position to complete the last piece of magic. That amounts to solving the following problem
\begin{eqnarray}
\min_{\z} \max_{\|\lambda\|_2=1}& & \z_{sc}f_{RD,2} \nonumber \\
\mbox{subject to}  & & \z\in \left [0,2/\sqrt{n}\right ]^n, \z_{sc}=\sqrt{\|\z\|_2^2+\sigma^2}.\label{eq:avoidsecclup1a2a13}
\end{eqnarray}
Relying again on the concentration concept discussed on many occasions in this and the previous section (as well as in \cite{Stojnicclupint19}) and earlier in \cite{StojnicCSetam09,StojnicISIT2010binary,StojnicDiscPercp13,StojnicGenLasso10,StojnicGenSocp10,StojnicPrDepSocp10,StojnicRegRndDlt10}) and connecting to (\ref{eq:avoidsecclup1a3})
\begin{eqnarray}
\min_{\z} \max_{\|\lambda\|_2=1}& & \z_{sc}f_{RD,2} \nonumber \\
\mbox{subject to}  & & \z\in \left [0,2/\sqrt{n}\right ]^n, \z_{sc}=\sqrt{\|\z\|_2^2+\sigma^2}\nonumber \\
& & \|\z\|_2^2=c_{2,z}\nonumber \\
& & (\x^{(1,s)})^T\z=s_2.\label{eq:avoidsecclup1a2a14}
\end{eqnarray}
The above is in principle the core of the mechanism. However, to ensure that we can track the behavior of all critical parameters we will also add another concentrating constraint $\x_{sol}^T\z=\1^T\z/\sqrt{n}=s_3$ ($\1$ is of course the $n$-dimensional column vector of all ones) so that we actually have
\begin{eqnarray}
\min_{\z} \max_{\|\lambda\|_2=1}& & \z_{sc}f_{RD,2} \nonumber \\
\mbox{subject to}  & & \z\in \left [0,2/\sqrt{n}\right ]^n, \z_{sc}=\sqrt{\|\z\|_2^2+\sigma^2}\nonumber \\
& & \|\z\|_2^2=c_{2,z}\nonumber \\
& & (\x^{(1,s)})^T\z=s_2 \nonumber \\
& & \frac{1}{\sqrt{n}}\1^T\z=s_3 \nonumber \\
& & \lambda^T\hat{\lambda}_0=p^{(1)},\label{eq:avoidsecclup1a2a15}
\end{eqnarray}
where $\hat{\lambda}_0=\g/\|\g\|_2$ is obtained trivially from (\ref{eq:avoidsecclup1a2a6}) and (\ref{eq:avoidsecclup1a2a10}). One should also keep in mind that
\begin{equation}\label{eq:avoidsecclup1a2a16}
  q^{(1)}=\begin{bmatrix}\z\\\sigma\end{bmatrix}^T\begin{bmatrix}\z^{(1)}\\\sigma\end{bmatrix}/\z_{sc}/\z_{sc}^{(1)}
  =\frac{s_3-s_2+\sigma^2}{\sqrt{c_{2,z}+\sigma^2}\sqrt{\hat{c}_{1,z}+\sigma^2}}.
\end{equation}
Finally, after plugging back the value for $f_{RD,2}$ from (\ref{eq:avoidsecclup1a2a5}) one has
\begin{eqnarray}
\min_{\z} \max_{\|\lambda\|_2=1}& & \z_{sc}\lambda^T(q^{(1)}\g+\sqrt{1-(q^{(1)})^2}\g^{(1)})+(p^{(1)}\h+\sqrt{1-(p^{(1)})^2}\h^{(1)})\z  \nonumber \\
\mbox{subject to}  & & \z\in \left [0,2/\sqrt{n}\right ]^n, \z_{sc}=\sqrt{\|\z\|_2^2+\sigma^2}\nonumber \\
& & \|\z\|_2^2=c_{2,z}\nonumber \\
& & (\x^{(1,s)})^T\z=s_2 \nonumber \\
& & \1^T\z=\sqrt{n}s_3 \nonumber \\
& & \lambda^T\g=\|\g\|_2p^{(1)},\label{eq:avoidsecclup1a2a17}
\end{eqnarray}
where the last term in $f_{RD,2}$ in (\ref{eq:avoidsecclup1a2a5}), $(p^{(1)}h_0+\sqrt{1-(p^{(1)})^2}h_0^{(1)})\sigma$, is neglected. Choosing
\begin{equation}\label{eq:avoidsecclup1a2a18}
  \hat{\lambda}=\frac{p^{(1)}\g+\sqrt{1-(p^{(1)})^2}\g^{(1)}}{\|p^{(1)}\g+\sqrt{1-(p^{(1)})^2}\g^{(1)}\|_2},
\end{equation}
and averaging over $\g$ and $\g^{(1)}$ we from (\ref{eq:avoidsecclup1a2a17}) then obtain
\begin{eqnarray}
\min_{\z} & & \sqrt{\alpha n}\sqrt{\|\z\|_2^2+\sigma^2}\left (q^{(1)}p^{(1)}+\sqrt{1-(q^{(1)})^2}\sqrt{1-(p^{(1)})^2}\right )+(p^{(1)}\h+\sqrt{1-(p^{(1)})^2}\h^{(1)})\z  \nonumber \\
\mbox{subject to}  & & \z\in \left [0,2/\sqrt{n}\right ]^n \nonumber \\
& & \|\z\|_2^2=c_{2,z}\nonumber \\
& & (\x^{(1,s)})^T\z=s_2 \nonumber \\
& & \1^T\z=\sqrt{n}s_3.\label{eq:avoidsecclup1a2a19}
\end{eqnarray}
After writing the Lagrange dual one obtains a problem very similar to (\ref{eq:avoidclup11})
\begin{eqnarray}
\max_{\gamma,\nu,\nu_2}\min_{\z} & & {\cal L}(\gamma,\nu,\nu_2)\nonumber \\
\mbox{subject to}  & & \z\in \left [0,2/\sqrt{n}\right ]^n,\label{eq:avoidsecclup11}
\end{eqnarray}
where
\begin{eqnarray}\label{eq:avoidsecclup11a}
{\cal L}(\gamma,\nu,\nu_2) & = & \sqrt{\alpha n}\sqrt{\|\z\|_2^2+\sigma^2}\left (q^{(1)}p^{(1)}+\sqrt{1-(q^{(1)})^2}\sqrt{1-(p^{(1)})^2}\right )\nonumber \\
& & +
\h^{(1,p)}\z +\gamma (\|\z\|_2^2-c_{2,z})+\nu ((\x^{(1,s)})^T\z-s_2+\nu_2(\1^T\z-\sqrt{n}s_3)),
\end{eqnarray}
with $\h^{(1,p)}=(p^{(1)}\h+\sqrt{1-(p^{(1)})^2}\h^{(1)})$. Similarly to (\ref{eq:avoidclup11}), we refer to the expected value of the $\sqrt{n}$ scaled version of the above objective as $\xi_{RD}^{(2)}(\alpha,\sigma;p^{(1)},q^{(1)},c_{2,z},s_2,s_3,\gamma,\nu,\nu_2)$. Then analogously to (\ref{eq:avoidclup12}) (and following into the footsteps of say \cite{StojnicDiscPercp13}) we define
\begin{eqnarray}
f_{box,2}(\h^{(1,p)};c_{2,z},s_2,s_3) = \max_{\gamma,\nu,\nu_2}\min_{\z}&& \h^{(1,p)}\z +\gamma (\|\z\|_2^2-c_{2,z})+\nu ((\x^{(1,s)})^T\z-s_2+\nu_2(\frac{1}{\sqrt{n}}\1^T\z-s_3))\nonumber \\
\mbox{subject to} & & \z\in \left [0,2/\sqrt{n}\right ]^n.\label{eq:avoidseccclup12}
\end{eqnarray}
The only thing that is different now compared to (\ref{eq:avoidclup12}) is that we now have an extra constraint related to $s_3$. This though changes nothing with respect to the methodology that we applied after (\ref{eq:avoidclup12}) and we could utilize the solution obtained there with a slight modification to account for $s_3$ and $\nu_2$. That essentially means that instead of (\ref{eq:avoidclup13}) (and earlier (110) from \cite{StojnicDiscPercp13}) one now has
\begin{eqnarray}
f_{box,2}(\h^{(1,p)};c_{2,z},s_2,s_3)  =  \max_{\gamma,\nu} & & \frac{1}{\sqrt{n}}\left (\sum_{i=1}^{n}f_{box,2}^{(1)}(\h_i^{(1,p)},\gamma,\nu)\right )-\nu s_2\sqrt{n}-\nu s_3\sqrt{n}-\gamma c_{2,z}\sqrt{n},\label{eq:avoidsecclup13}
\end{eqnarray}
where
\begin{equation}
f_{box,2}^{(1)}(\h_i^{(1,p)},\gamma,\nu,\nu_2)=\begin{cases}0, & \h_i^{(1,p)}+\nu\x^{(0)}_i\geq 0\\
-\frac{(\h_i^{(1,p)}+\nu\x^{(1,s)}_i+\nu_2)^2}{4\gamma}, & -4\gamma\leq \h_i^{(1,p)}+\nu\x^{(1,s)}_i+\nu_2\leq 0\\
2(\h_i^{(1,p)}+\nu\x^{(1,s)}_i+\nu_2)+4\gamma, & \h_i^{(1,p)}+\nu\x^{(1,s)}_i+\nu_2\leq -4\gamma,
\end{cases}\label{eq:avoidsecclup14}
\end{equation}
with the same scaling discussion regarding $\gamma$, $\nu$, $\nu_2$, and $\x^{(0)}_i$ as after (\ref{eq:avoidclup14}). Following further what was done in the previous section (and earlier outlined in \cite{Stojnicclupint19}) one can also determine the optimizing $\z_i$ and $\x_i^{(2,s)}$
\begin{eqnarray}
\z_i^{(2)} & = & \frac{1}{\sqrt{n}}\min \left (\max\left (0,-\left (\frac{\h_i^{(1,p)}+\nu\x^{(1,s)}_i+\nu_2}{2\gamma}\right )\right ),2\right )\nonumber \\
\x_i^{(2,s)} & = & \frac{1}{\sqrt{n}}-\z_i^{(2)}=\frac{1}{\sqrt{n}}\left (1-\min \left (\max\left (0,-\left (\frac{\h_i^{(1,p)}+\nu\x^{(1,s)}_i+\nu_2}{2\gamma}\right )\right ),2\right )\right ),\label{eq:avoidsecclup14a}
\end{eqnarray}
where we also recall from (\ref{eq:avoidsecclup1a2a12})
\begin{equation}
\x_i^{(1,s)}=1-\z_i^{(1)}=\frac{1}{\sqrt{n}}\left (1-\min \left (\max\left (0,-\left (\frac{\h+\hat{\nu}^{(1)}\x^{(0)}_i}{2\hat{\gamma}^{(1)}}\right )\right ),2\right )\right ).\label{eq:avoidsecclup15a}
\end{equation}
As earlier, if we assume that the initial $\x^{(0)}$ has $\rho n$ components equal to $\frac{1}{\sqrt{n}}$ and $(1-\rho) n$ components equal to $-\frac{1}{\sqrt{n}}$ we have for the objective
\begin{equation}
\mE f_{box,2}^{(1)}(\h_i,\h_i^{(1)},\gamma,\nu,\nu_2)=\rho I_{1}^{(2)}(\gamma,\nu,\nu_2,\hat{\nu}^{(1)},\hat{\gamma}^{(1)})+(1-\rho) I_{1}^{(2)}(\gamma,\nu,\nu_2,-\hat{\nu}^{(1)},\hat{\gamma}^{(1)}),\label{eq:avoidsecclup15}
\end{equation}
where
\begin{equation}
I_{1}^{(2)}(\gamma,\nu,\nu_2,\hat{\nu}^{(1)},\hat{\gamma}^{(1)})  =  \int\int((\h_i^{(1,p)}+\nu\x^{(1,s)}+\nu_2)\z_i^{(2)}+\gamma\left (\z_i^{(2)}\right )^2)exp\left (-\frac{\left (\h_i^{(1)}\right )^2+ \h_i^2}{2}\right )\frac{d\h_i^{(1)}d\h_i}{2\pi},
\label{eq:avoidsecclup16}
\end{equation}
and for $\gamma<0$ the term under the integral is replaced by zero if negative. Combining all of the above finally gives
\begin{eqnarray}
\xi_{RD}^{(2)}(\alpha,\sigma;p^{(1)},q^{(1)},c_{2,z},s_2,s_3,\gamma,\nu,\nu_2) & = & \sqrt{\alpha}\sqrt{c_{2,z}+\sigma^2}\left (q^{(1)}p^{(1)}+\sqrt{1-(q^{(1)})^2}\sqrt{1-(p^{(1)})^2}\right )\nonumber \\
& & +\mE f_{box,2}^{(1)}(\h_i,\h_i^{(1)},\gamma,\nu,\nu_2)-\nu s_2-\nu_2 s_3-\gamma c_{2,z}. \label{eq:avoidsecclup17}
\end{eqnarray}
Now we are in position to give a sort of a brief summary of the entire formalism that we presented above. It is essentially analogous to what we discussed after Theorem \ref{thm:avoidcluprd1}. While we will try to emulate all the ideas from the previous section, there are quite a few new elements here that need to be incorporated and we will try to emphasize all of that in the formalism below.

\subsection{CLuP -- summary of the second iteration performance analysis}
\label{sec:clupseconditsummary}

Similarly to what we did in Section \ref{sec:clupfirstitsummary}, we below present critical steps needed to actually calculate all the quantities of interest. Of course, the above analysis is at the core of all the underlying mechanisms that basically enable performing these steps.

\vspace{.1in}
\noindent \xmyboxc{\bl{\emph{\textbf{Summarized formalism to handle the CLuP's second iteration }}}}

Differently from Section \ref{sec:clupfirstitsummary}, here we will split the presentation into several separate parts.

\vspace{.1in}
\noindent \xmyboxc{\bl{\emph{\textbf{I) First part -- \dgr{Handling the first iteration} }}}}

The first part of the formalism essentially reflects on the above analysis by recognizing that what it effectively accomplished was rehandling the CLuP's first iteration. That basically means that one first solves the following problem
\begin{eqnarray}
\phi_a^{(1)}=\{\hat{\nu}^{(1)},\hat{\gamma}^{(1)},\hat{c}_{1,z}^{(1)},\hat{s}_1^{(1)}\}=\mbox{arg} \min_{s_1} & & s_1 \nonumber \\
\mbox{subject to} & & \min_{0\leq c_{1,z}\leq 4}\max_{\gamma,\nu}\xi_{RD}^{(1)}(\alpha,\sigma;c_{1,z},s_1,\gamma,\nu)= r.\label{eq:sumavoidclup17a}
\end{eqnarray}
to obtain set of parameters $\{\hat{\nu}^{(1)},\hat{\gamma}^{(1)},\hat{c}_{1,z}^{(1)},\hat{s}_1^{(1)}\}$ that enter the second iteration. Recalling on (\ref{eq:avoidclup18})
\begin{eqnarray}
s_{x,1}(\gamma,\nu) & = & -\nu/2/\gamma(.5\erfc(\nu/\sqrt{2})-.5\erfc((\nu+4\gamma)/\sqrt{2}))\nonumber \\
& & +1/2/\gamma/\sqrt{2\pi}(exp(-\nu^2/2)-exp(-(4\gamma+\nu)^2/2))\nonumber \\
s_{xsq,1}(\gamma,\nu) & = & -I_{1,1}(\gamma,\nu)/\gamma\nonumber \\
s_{x,2}(\gamma,\nu) & = & 2(.5\erfc((4\gamma+\nu)/\sqrt{2}))\nonumber \\
s_{xsq,2}(\gamma,\nu) & = & 2s_{x,2}.\label{eq:sumavoidclup18}
\end{eqnarray}
one then has from (\ref{eq:avoidclup21}) the first iteration values for the first \textbf{two critical parameters} related to the propagation of the vector $\x$ through the CLuP algorithm that we particularly keep track of
\begin{eqnarray}
\hat{d}_1^{(1)}\triangleq \mE((\x_{sol})^T\x^{(1,s)}) \hspace{-.08in}& = & \hspace{-.08in}1-(\rho s_{x,1}(\hat{\gamma}^{(1)},\hat{\nu}^{(1)})+(1-\rho)s_{x,1}(\hat{\gamma}^{(1)},-\hat{\nu}^{(1)})\nonumber \\
& & +\rho s_{x,2}(\hat{\gamma}^{(1)},\hat{\nu}^{(1)})+(1-\rho)s_{x,2}(\hat{\gamma}^{(1)},-\hat{\nu}^{(1)}))\nonumber \\
\hat{d}_2^{(1)}\triangleq \mE\|\x^{(1,s)}\|_2^2 \hspace{-.08in} & = & \hspace{-.08in} \rho s_{xsq,1}(\hat{\gamma}^{(1)},\hat{\nu}^{(1)})+(1-\rho)s_{xsq,1}(\hat{\gamma}^{(1)},-\hat{\nu}^{(1)})\nonumber \\
&&+\rho s_{xsq,2}(\hat{\gamma}^{(1)},\hat{\nu}^{(1)})+(1-\rho)s_{xsq,2}(\hat{\gamma}^{(1)},-\hat{\nu}^{(1)})
+2\mE((\x_{sol})^T\x^{(1,s)})-1.\label{eq:sumavoidclup21}
\end{eqnarray}
Moreover, from (\ref{eq:avoidsecclup1a2a11}) and (\ref{eq:avoidsecclup1a2a12}) we also have
\begin{equation}
\z_i^{(1)}=\frac{1}{\sqrt{n}}\min \left (\max\left (0,-\left (\frac{\h_i+\hat{\nu}\x^{(0)}_i}{2\hat{\gamma}}\right )\right ),2\right ),\label{eq:sumavoidsecclup1a2a11}
\end{equation}
and
\begin{equation}
\x_i^{(1,s)}=1-\z_i^{(1)}=\frac{1}{\sqrt{n}}\left (1-\min \left (\max\left (0,-\left (\frac{\h+\hat{\nu}^{(1)}\x^{(0)}_i}{2\hat{\gamma}^{(1)}}\right )\right ),2\right )\right ),\label{eq:sumavoidsecclup1a2a12}
\end{equation}
and from (\ref{eq:avoidclup21a}) the \textbf{third critical parameter} that we keep a track of through iterations, the probability of error
\begin{equation}\label{eq:sumavoidclup21a}
  p_{err}^{(1)}=1-P\left (\z_i^{(1)}\leq \frac{1}{\sqrt{n}}\right )=1-\left (\rho\left (\frac{1}{2}\erfc\left ( \frac{-2\hat{\gamma}^{(1)}-\hat{\nu}^{(1)}}{\sqrt{2}}\right )\right )+(1-\rho)\left ( \frac{1}{2}\erfc\left ( \frac{-2\hat{\gamma}^{(1)}+\hat{\nu}^{(1)}}{\sqrt{2}}\right )\right )\right ).
\end{equation}
The \textbf{fourth critical parameter} is of course the value of the objective and after a cosmetic change it is $\hat{s}^{(1)}\triangleq\hat{s}_1^{(1)}$. One can then basically say that in addition to the solution after the first iteration, $\x^{(1,s)}$, the following set of \bl{critical} plus \prp{auxiliary} parameters is the output of the first iteration:
\begin{equation}\label{eq:sumavoidclup21aa}
  \phi^{(1)}=\{\bl{p_{err}^{(1)},\hat{s}^{(1)},\hat{d}_2^{(1)},\hat{d}_1^{(1)}},\prp{\hat{\nu}^{(1)},\hat{\gamma}^{(1)},\hat{c}_{1,z}^{(1)}}\},
\end{equation}
where for simplicity we also emphasize in wording
\begin{eqnarray}\label{eq:sumavoidclup21aa1}
\bl{p_{err}^{(1)}} & - &  \mbox{probability of error after the first iteration}\nonumber \\
\bl{\hat{s}^{(1)}} & = &  \mE((\x^{(0)})^T\x^{(1,s)}) - \mbox{objective value after the first iteration} \nonumber \\
\bl{\hat{d}_2^{(1)}} & = & \mE\|\x^{(1,s)}\|_2^2 - \mbox{squared norm after the first iteration} \nonumber \\
\bl{\hat{d}_1^{(1)}} & = & \mE\x_{sol}^T\x^{(1,s)} - \mbox{inner product with $\x_{sol}$ after the first iteration}.\end{eqnarray}
It is probably not necessary to reemphasize but for the completeness we add that the last three quantities are viewed as expected values and since we are interested in large dimensional scenarios they are due to overwhelming concentrations also the concentrating points.

\vspace{.1in}
\noindent \xmyboxc{\bl{\emph{\textbf{II) Second part -- \dgr{Handling the second iteration} }}}}

Once the first iteration is handled one utilizes its parameters to basically run the second iteration. The strategy is conceptually to a degree similar to what we presented above for the first iteration. One starts with first solving the following optimization problem
\begin{eqnarray}
\phi_a^{(2)}=\mbox{arg} \min_{s_2,s_3} & & \frac{s_2-\hat{d}_1^{(1)}}{\sqrt{\hat{d}_2^{(1)}}}\nonumber \\
\mbox{subject to} & & \min_{q^{(1)}}\max_{p^{(1)}}\min_{0\leq c_{2,z}\leq 4}\max_{\gamma,\nu,\nu_2}\xi_{RD}^{(2)}(\alpha,\sigma;p^{(1)},q^{(1)},c_{2,z},s_2,s_3,\gamma,\nu,\nu_2)=r,\label{eq:sumsecavoidclup17a}
\end{eqnarray}
where
\begin{equation}\label{eq:sumsecavoidclup17aa}
\phi_a^{(2)}=\{\hat{p}^{(1)},\hat{q}^{(1)},\hat{\nu}^{(2)},\hat{\nu}_2^{(2)},\hat{\gamma}^{(2)},\hat{c}_{2,z}^{(2)},\hat{s}_2^{(2)},\hat{s}_3^{(2)}\}.
\end{equation}
With a few exceptions the above seems rather natural extension of (\ref{eq:sumavoidclup17a}). The main changes are the readjustment of the objective and the appearance of $p^{(1)}$ and $q^{(1)}$. Besides this the strategy in (\ref{eq:sumsecavoidclup17a}) essentially remains the same as in (\ref{eq:sumavoidclup17a}). This practically means that one wants to minimize the objective while keeping the optimized $\xi_{RD}^{(2)}(\alpha,\sigma;c_{2,z},s_2,s_3,\gamma,\nu,\nu_2)$ below $r$. The difference is that now the objective is not the $s_1$ as we had before but a rather different object that we explain below.

Now, looking carefully at (\ref{eq:avoidsecclup1})-(\ref{eq:avoidsecclup1a3}) and everything that followed afterwards, one can observe that instead of the real objective $-(\x^{(1)})^T(\x_{sol}-\z)$ its a bit more convenient version $s_2=(\x^{(1,s)})^T\z$ was utilized. To readjust for this we simply note that from (\ref{eq:sumavoidclup21}) one easily has
\begin{equation}\label{{eq:sumsecavoidclup17a1}}
  \frac{-(\x^{(1,s)})^T(\x_{sol}-\z)}{\|\x^{(1,s)}\|_2}=\frac{-\hat{d}_1^{(1)}+s_2}{\sqrt{\hat{d}_2^{(1)}}}.
\end{equation}
To be a bit more in alignment with what was done in the first part (in particular in (\ref{eq:sumavoidclup17a})) one may rewrite (\ref{eq:sumsecavoidclup17a}) in the following way
\begin{eqnarray}
\phi_a^{(2)}=\mbox{arg} \min_{s,s_2,s_3} & & s\nonumber \\
\mbox{subject to} & & \min_{q^{(1)}}\max_{p^{(1)}}\min_{0\leq c_{2,z}\leq 4}\max_{\gamma,\nu,\nu_2}\xi_{RD}^{(2)}(\alpha,\sigma;p^{(1)},q^{(1)},c_{2,z},s_2,s_3,\gamma,\nu,\nu_2)=r\nonumber \\
& & s_2=\hat{d}_1^{(1)}+s\sqrt{\hat{d}_2^{(1)}}.\label{eq:sumsecavoidclup17a2}
\end{eqnarray}
Now, carefully observing further (\ref{eq:sumsecavoidclup17a2}) one can also note that what it basically does is that instead of a parameter $s_2$ (which is natural to the above discussion), it actually reintroduces parameter $s$ (the value of the objective) as a probably more natural object for the following of the algorithm's flow. One can actually continue that way with the second iteration analogues to the other two critical parameters that we mentioned above in the summary of the first iteration's formalism. Namely, if one defines analogously to (\ref{eq:sumavoidclup21})
\begin{eqnarray}
\hat{d}_1^{(2)}& \triangleq &\mE((\x_{sol})^T\x^{(2,s)})=1-\hat{s}_3^{(2)}\nonumber \\
\hat{d}_2^{(2)}& \triangleq & \mE\|\x^{(2,s)}\|_2^2  = \hat{c}_{2,z}^{(2)}
+2\mE((\x_{sol})^T\x^{(2,s)})-1,\label{eq:sumsecavoidclup21}
\end{eqnarray}
then (\ref{eq:sumsecavoidclup17a2}) can be repositioned as
\begin{eqnarray}
\phi_b^{(2)}=\mbox{arg} \min_{s,d_1^{(2)},d_2^{(2)}} & & s\nonumber \\
\mbox{subject to} & & \min_{q^{(1)}}\max_{p^{(1)}}\min_{0\leq c_{2,z}\leq 4}\max_{\gamma,\nu,\nu_2}\xi_{RD}^{(2)}(\alpha,\sigma;p^{(1)},q^{(1)},c_{2,z},s_2,s_3,\gamma,\nu,\nu_2)=r\nonumber \\
& & s_2=d_1^{(1)}+s\sqrt{d_2^{(1)}}\nonumber \\
& & s_3=1-d_1^{(2)}\nonumber \\
& & c_{2,z}=d_2^{(2)}-2d_1^{(2)}+1,\label{eq:sumsecavoidclup17a3}
\end{eqnarray}
where
\begin{equation}\label{eq:sumsecavoidclup17a4}
\phi_b^{(2)}=\{\hat{p}^{(1)},\hat{q}^{(1)},\hat{\nu}^{(2)},\hat{\nu}_2^{(2)},\hat{\gamma}^{(2)},\hat{s}^{(2)},\hat{d}_2^{(2)},\hat{d}_1^{(2)}\}.
\end{equation}
Finally, if one for a moment recalls (\ref{eq:avoidsecclup1a2a16}) then (\ref{eq:sumsecavoidclup17a3}) can also be rewritten as
\begin{eqnarray}
\phi_b^{(2)}=\mbox{arg} \min_{s,d_1^{(2)},d_2^{(2)}} & & s\nonumber \\
\mbox{subject to} & & \max_{p^{(1)}}\min_{0\leq c_{2,z}\leq 4}\max_{\gamma,\nu,\nu_2}\xi_{RD}^{(2)}(\alpha,\sigma;p^{(1)},q^{(1)},c_{2,z},s_2,s_3,\gamma,\nu,\nu_2)=r\nonumber \\
& & s_2=d_1^{(1)}+s\sqrt{d_2^{(1)}}\nonumber \\
& & s_3=1-d_1^{(2)}\nonumber \\
& & c_{2,z}=d_2^{(2)}-2d_1^{(2)}+1\nonumber \\
& & q^{(1)} =\frac{s_3-s_2+\sigma^2}{\sqrt{c_{2,z}+\sigma^2}\sqrt{\hat{c}_{1,z}+\sigma^2}}.\label{eq:sumsecavoidclup17a5}
\end{eqnarray}
We also recall that $\xi_{RD}^{(2)}(\alpha,\sigma;p^{(1)},q^{(1)},c_{2,z},s_2,s_3,\gamma,\nu,\nu_2)$ is given through (\ref{eq:avoidsecclup15a})-(\ref{eq:avoidsecclup17}) and that is where the remaining auxiliary parameters from the first iteration, $\hat{\nu}^{(1)}$ and $\hat{\gamma}^{(1)}$, come into the play as well.

Similarly to (\ref{eq:sumavoidclup21a}) one also has
\begin{equation}\label{eq:sumsecavoidclup17a6}
  p_{err}^{(2)}=1-(\rho p_{cor}(\hat{\nu}^{(1)})+(1-\rho) p_{cor}(-\hat{\nu}^{(1)})),
\end{equation}
where
\begin{equation}\label{eq:sumsecavoidclup17a7}
 p_{cor}(\hat{\nu}^{(1)})=\int\int((\mbox{sign}(\x^{(2,s)})+1)/2)exp\left (-\frac{\left (\h_i^{(1)}\right )^2+ \h_i^2}{2}\right )\frac{d\h_i^{(1)}d\h_i}{2\pi}.
\end{equation}
For the completeness we also add that
\begin{eqnarray}\label{eq:sumsecavoidclup17a8}
 \hat{d}_{2,+}^{(2)}(\hat{\nu}^{(1)}) & = & \int\int((\x_i^{(2,s)})^2exp\left (-\frac{\left (\h_i^{(1)}\right )^2+ \h_i^2}{2}\right )\frac{d\h_i^{(1)}d\h_i}{2\pi}\nonumber \\
  \hat{d}_{1,+}^{(2)}(\hat{\nu}^{(1)}) & = & \int\int((\x_i^{(2,s)})exp\left (-\frac{\left (\h_i^{(1)}\right )^2+ \h_i^2}{2}\right )\frac{d\h_i^{(1)}d\h_i}{2\pi}\nonumber \\
    \hat{s}_{2,+}^{(2)}(\hat{\nu}^{(1)}) & = & \int\int((\x_i^{(1,s)})\z_i^{(2)}exp\left (-\frac{\left (\h_i^{(1)}\right )^2+ \h_i^2}{2}\right )\frac{d\h_i^{(1)}d\h_i}{2\pi},
\end{eqnarray}
and
\begin{eqnarray}\label{eq:sumsecavoidclup17a9}
\hat{d}_{2}^{(2)}  & = & \rho\hat{d}_{2,+}^{(2)}(\hat{\nu}^{(1)})+(1-\rho)\hat{d}_{2,+}^{(2)}(-\hat{\nu}^{(1)}) \nonumber \\
\hat{d}_{1}^{(2)}  & = & \rho\hat{d}_{1,+}^{(2)}(\hat{\nu}^{(1)})+(1-\rho)\hat{d}_{1,+}^{(2)}(-\hat{\nu}^{(1)}) \nonumber \\
\hat{s}_{2}^{(2)}  & = & \rho\hat{s}_{2,+}^{(2)}(\hat{\nu}^{(1)})+(1-\rho)\hat{s}_{2,+}^{(2)}(-\hat{\nu}^{(1)}).
\end{eqnarray}

Finally, similarly to the end of the summary of the first part, here one can also say that in addition to the solution after the second iteration, $\x^{(2,s)}$, the following set of \bl{critical} plus \prp{auxiliary} parameters is the output of the second iteration:
\begin{equation}\label{eq:sumavoidclup21aa}
  \phi^{(2)}=\{\bl{p_{err}^{(2)},\hat{s}^{(2)},\hat{d}_2^{(2)},\hat{d}_1^{(2)}},\prp{\hat{\nu}^{(2)},\hat{\nu}_2^{(2)},\hat{\gamma}^{(2)},\hat{p}^{(1)},\hat{q}^{(1)},\hat{c}_{2,z}^{(1)},\hat{s}_{2}^{(2)},\hat{s}_{3}^{(2)}}\},
\end{equation}
where we again for simplicity use the wording to emphasize
\begin{eqnarray}\label{eq:sumavoidclup21aa1}
\bl{p_{err}^{(2)}} & - &  \mbox{probability of error after the second iteration}\nonumber \\
\bl{\hat{s}^{(2)}} & = &  \mE((\x^{(1)})^T\x^{(2,s)}) - \mbox{objective value after the second iteration} \nonumber \\
\bl{\hat{d}_2^{(2)}} & = & \mE\|\x^{(2,s)}\|_2^2 - \mbox{squared norm after the second iteration} \nonumber \\
\bl{\hat{d}_1^{(2)}} & = & \mE\x_{sol}^T\x^{(2,s)} - \mbox{inner product with $\x_{sol}$ after the second iteration}.
\end{eqnarray}
It goes again without much of a discussion that the last three quantities are viewed as expected/concentrating values.

\subsection{Numerical results -- second iteration}
\label{sec:clupsecondnum}

To follow into the footsteps of the discussion regarding the analysis of the CLuP's first iteration, we in this section present a set of numerical results that relate to the above analysis, in particular to the CLuP's second iteration. As in Section \ref{sec:clupfirstnum} a numerical analysis is needed for both, the theoretical and simulated values. We again start with the theoretical predictions. To that end, we first actually recall that the input to the analysis of the second iteration is
\begin{eqnarray}\label{eq:numressec1}
  \phi^{(1)}  =  \{\bl{p_{err}^{(1)},\hat{s}^{(1)},\hat{d}_2^{(1)},\hat{d}_1^{(1)}},\prp{\hat{\nu}^{(1)},\hat{\gamma}^{(1)},\hat{c}_{1,z}^{(1)}}\}=  \{0.0456,-0.1314,0.7009,0.7628,0.4953,0.9420,0.1753\}.
\end{eqnarray}
From (\ref{eq:sumavoidclup21aa}) for the second iteration's output set of parameters we have
\begin{equation}\label{eq:numressec2}
  \phi^{(2)}=\{\bl{p_{err}^{(2)},\hat{s}^{(2)},\hat{d}_2^{(2)},\hat{d}_1^{(2)}},\prp{\hat{\nu}^{(2)},\hat{\nu}_2^{(2)},\hat{\gamma}^{(2)},\hat{p}^{(1)},\hat{q}^{(1)},\hat{c}_{2,z}^{(1)},\hat{s}_{2}^{(2)},\hat{s}_{3}^{(2)}}\}.
\end{equation}
We show the theoretical values for some of the system parameters obtained based on the above analysis for SNR, $1/\sigma^2=13$[db], $\alpha=0.8$, and $r_{sc}=1.3$ in Table \ref{tab:tabsecavoidnum1}.
\begin{table}[h]
\caption{\textbf{Theoretical} values for various system parameters obtained utilizing the above analysis}\vspace{.1in}
\hspace{-0in}\centering
\small{
\begin{tabular}{||c||c|c|c||c||c|c|c|c|c||}\hline\hline
$1/\sigma^2 $[db]  & $\hat{\nu}^{(2)}$ & $\hat{\nu}_2^{(2)}$ & $\hat{\gamma}^{(2)}$ & $\hat{p}^{(1)}$ & $\hat{s}^{(2)}$ &   $\xi_{RD}^{(2)} $ & $p_{err}^{(2)} $  & $\|\x^{(2,s)}\|_2^2$ &
$(\x_{sol})^T\x^{(2,s)}$ \\ \hline\hline
$13  $ & $\mathbf{2.6924}  $ & $\mathbf{-0.6428}  $ & $\mathbf{1.8911}  $ & $\mathbf{0.70}  $ & $\mathbf{-0.9117}  $ & $\mathbf{0.1594}  $ & $\mathbf{0.00651}  $ & $\mathbf{0.9064}  $ & $\mathbf{0.9340}  $ \\ \hline\hline
\end{tabular}}
\label{tab:tabsecavoidnum1}
\end{table}
Utilizing the equality constraints in (\ref{eq:sumsecavoidclup17a5}) one can also easily obtain the remaining parameters from $\phi^{(2)}$, i.e. $\prp{\{\hat{q}^{(1)},\hat{c}_{2,z}^{(1)},\hat{s}_{2}^{(2)},\hat{s}_{3}^{(2)}}\}$. In Table \ref{tab:tabsecavoidnum2} we show the results obtained utilizing both, the equality constraints in (\ref{eq:sumsecavoidclup17a5}) as well as (\ref{eq:sumsecavoidclup17a8}) and (\ref{eq:sumsecavoidclup17a9}).
\begin{table}[h]
\caption{\textbf{Theoretical} values for $\prp{\{\hat{q}^{(1)},\hat{c}_{2,z}^{(1)},\hat{s}_{2}^{(2)},\hat{s}_{3}^{(2)}}\}$ obtained utilizing (\ref{eq:sumsecavoidclup17a5}) (\textbf{bold}) as well as (\ref{eq:sumsecavoidclup17a8}) and (\ref{eq:sumsecavoidclup17a9}) (\prp{\textbf{purple bold}}) }\vspace{.1in}
\hspace{-0in}\centering
\small{
\begin{tabular}{||c||c|c|c|c||}\hline\hline
$1/\sigma^2 $[db]  & $\hat{s}_2^{(2)}=\hat{d}_1^{(1)}+\hat{s}^{(2)}\sqrt{\hat{d}_2^{(1)}}$ & $\hat{s}_3^{(2)}=1-\hat{d}_1^{(2)}$ & $\hat{c}_{2,z}^{(2)}=\hat{d}_2^{(2)}-2\hat{d}_1^{(2)}+1$ & $\hat{q}^{(1)}=\frac{\hat{s}_3-\hat{s}_2+\sigma^2}{\sqrt{\hat{c}_{2,z}+\sigma^2}\sqrt{\hat{c}_{1,z}+\sigma^2}}$ \\ \hline\hline
$13  $ & $\mathbf{-0.000458}  $/\prp{$\mathbf{-0.000458}  $} & $\mathbf{0.0660}  $/\prp{$\mathbf{0.0660}  $}&
$\mathbf{0.0384}$/\prp{$\mathbf{0.0384}$} & $\mathbf{0.8253}$/\prp{$\mathbf{0.8253}  $}    \\ \hline\hline
\end{tabular}}
\label{tab:tabsecavoidnum2}
\end{table}

\subsubsection{Simulations -- second iteration}
\label{sec:clupsecondsim}

Similarly to what we did in Section \ref{sec:clupfirstsim}, we below provide a set of results obtained through numerical simulations. In Table \ref{tab:tabsecavoidnum3} we show the simulated values that correspond to the above theoretical predictions. We obtained these values for $\alpha=0.8$, $n=1600$, and $r_{sc}=1.3$. As earlier, since $r_{plt}=.1226$ one easily has $\xi_{RD}^{(2)}=r=r_{sc}r_{plt}=.1594$.
\begin{table}[h]
\caption{\textbf{Theoretical}/\bl{\textbf{simulated}} values for various system parameters obtained based on the above analysis}\vspace{.1in}
\hspace{-0in}\centering
\small{
\begin{tabular}{||c||c||c|c|c|c||}\hline\hline
$1/\sigma^2 $[db]  & $\hat{s}^{(2)}$ &   $\xi_{RD}^{(2)} $ & $p_{err}^{(2)} $  & $\|\x^{(2,s)}\|_2^2$ &
$(\x_{sol})^T\x^{(2,s)}$ \\ \hline\hline
$13  $ &  $\mathbf{-0.9117}  $/\bl{$\mathbf{-0.9123}  $} & $\mathbf{0.1594}  $/\bl{$\mathbf{0.1594}  $} & $\mathbf{0.0065}  $/\bl{$\mathbf{0.0072}  $} & $\mathbf{0.9064}  $/\bl{$\mathbf{0.9061}  $} & $\mathbf{0.9340}  $/\bl{$\mathbf{0.9332}  $} \\ \hline\hline
\end{tabular}}
\label{tab:tabsecavoidnum3}
\end{table}
Following further what we did in Section \ref{sec:clupfirstsim}, we below in Table \ref{tab:tabsecavoidnum4} show what kind of effect the change of the problem dimension $n$ has on all the key parameters. The simulated values are again very close to the theoretical predictions. In particular, already when $n=1600$ all parameters are almost equal to the theoretical estimates. One can also observe that as $n$ grows almost all parameters are getting closer to the theoretical values.
\begin{table}[h]
\caption{\bl{\textbf{Simulated}} values for $p_{err}^{(2)}  $, $-\hat{s}^{(2)}  $, $\|\x^{(2,s)}\|_2^2$, and $(\x_{sol})^T\x^{(2,s)}$; $\alpha=0.8$, $r_{sc}=1.3$, $n=\{100,200,400,800,1600\}$}\vspace{.1in}
\hspace{-0in}\centering
\footnotesize{
\begin{tabular}{||c||c||c|c|c|c||}\hline\hline
$ n $ & $\#$ of reps.   & $p_{err}^{(2)}  $ & $-\hat{s}^{(2)}  $ & $\hat{d}_2^{(2)}=\|\x^{(2,s)}\|_2^2  $ & $\hat{d}_1^{(2)}=(\x_{sol})^T\x^{(2,s)}  $  \\ \hline\hline
$100  $ & $165  $ & $\bl{\mathbf{0.0188}}  $ & $\bl{\mathbf{0.9098}}  $ & $\bl{\mathbf{0.8886}}  $ & $\bl{\mathbf{0.9081}}  $ \\ \hline
$200  $ & $954  $ & $\bl{\mathbf{0.0158}}  $ & $\bl{\mathbf{0.9096}}  $ & $\bl{\mathbf{0.8933}}  $ & $\bl{\mathbf{0.9155}}  $ \\ \hline
$400  $ & $600  $ & $\bl{\mathbf{0.0111}}  $ & $\bl{\mathbf{0.9097}}  $ & $\bl{\mathbf{0.8980}}  $ & $\bl{\mathbf{0.9239}}  $ \\ \hline
$800  $ & $465  $ & $\bl{\mathbf{0.0079}}  $ & $\bl{\mathbf{0.9126}}  $ & $\bl{\mathbf{0.9048}}  $ & $\bl{\mathbf{0.9317}}  $ \\ \hline
$1600  $ & $170  $ & $\bl{\mathbf{0.0072}}  $ & $\bl{\mathbf{0.9123}}  $ & $\bl{\mathbf{0.9061}}  $ & $\bl{\mathbf{0.9332}}  $ \\ \hline\hline
$\infty$ -- \textbf{theory} & $-$ & $\mathbf{0.0065}$  & $\mathbf{0.9117}$ & $\mathbf{0.9064}$ & $\mathbf{0.9340}$ \\ \hline\hline
\end{tabular}}
\label{tab:tabsecavoidnum4}
\end{table}
Finally, in Table \ref{tab:tabsecavoidnum5} we show the progress through the first two iterations of all the critical parameters and their simulated values (as in Table \ref{tab:tabsecavoidnum3}, $\alpha=0.8$, $r_{sc}=1.3$, and $n=1600$).
\begin{table}[h]
\caption{\textbf{Theoretical}/\bl{\textbf{simulated}} values for key system parameters through the first two iterations}\vspace{.1in}
\hspace{-0in}\centering
\small{
\begin{tabular}{||c||c||c|c|c|c|c||}\hline\hline
$k$ & $1/\sigma^2 $[db]  & $-\hat{s}^{(k)}$ &   $\xi_{RD}^{(k)} $ & $p_{err}^{(k)} $  & $\hat{d}_2^{(k)}=\|\x^{(k,s)}\|_2^2$ &
$\hat{d}_1^{(k)}=(\x_{sol})^T\x^{(k,s)}$ \\ \hline\hline
$1$ & $13  $ & $\mathbf{0.1314  }/\bl{\mathbf{0.1308  }}$ & $\mathbf{0.1594  }/\bl{\mathbf{0.1594  }}$ & $\mathbf{0.0456  }/\bl{\mathbf{0.0454  }}$ & $\mathbf{0.7009  }/\bl{\mathbf{0.7028  }}$ & $\mathbf{0.7628  }/\bl{\mathbf{0.7642  }}$ \\ \hline\hline
$2$ & $13  $ &  $\mathbf{0.9117}  $/\bl{$\mathbf{0.9123}  $} & $\mathbf{0.1594}  $/\bl{$\mathbf{0.1594}  $} & $\mathbf{0.0065}  $/\bl{$\mathbf{0.0072}  $} & $\mathbf{0.9064}  $/\bl{$\mathbf{0.9061}  $} & $\mathbf{0.9340}  $/\bl{$\mathbf{0.9332}  $} \\ \hline\hline
\end{tabular}}
\label{tab:tabsecavoidnum5}
\end{table}

\section{CLuP -- $(k+1)$-th iteration performance analysis}
\label{sec:clupkit}

Given that the above discussion demonstrated that one can handle the algorithm's first two iterations one naturally wonders can it be extended so that it eventually covers all iterations. The answer is yes. That is indeed in principle possible. Moreover, not much more needs to be added to the already introduced technical/strategic components of the analysis. However, the number of the running parameters starts to rapidly increase. That doesn't change much when it comes to conceptual handling all of them. What does become affected though are the numerical evaluations. We will below briefly sketch how one can extend the above analysis and then show a couple of shortcuts regarding the numerical evaluations. In the first part we will essentially closely follow the presentation of the previous section. Instead of discussing all the details we will focus on the final results.

We of course start by restating the CLuP's $k+1$-th iteration underlying optimization problem. Namely, it boils down to finding $\x^{(k+1)}$ in the following way
\begin{eqnarray}
\x^{(k+1)}=\frac{\x^{(k+1,s)}}{\|\x^{(k+1,s)}\|_2} \quad \mbox{with}\quad \x^{(k+1,s)}=\mbox{arg}\min_{\x} & & -(\x^{(k)})^T\x  \nonumber \\
\mbox{subject to} & & \|\y-A\x\|_2\leq r\nonumber \\
&& \x\in \left [-\frac{1}{\sqrt{n}},\frac{1}{\sqrt{n}}\right ]^n. \label{eq:kit0}
\end{eqnarray}
To handle this problem we of course rely on RDT and what we presented in Section \ref{sec:clupsecondit}.

\vspace{.1in}
\noindent \xmyboxc{\bl{\emph{\textbf{1. First step -- \dgr{Forming the deterministic Lagrange dual} }}}}

Following closely what was done in Section \ref{sec:clupsecondit} one can arrive to the following analogue of (\ref{eq:avoidsecclup1a2a3})
\begin{eqnarray}
\max_{\gamma,\gamma_{j-1}\geq 0}\min_{\z,\z^{(j)}} \max_{\|\lambda\|_2=1,\|\lambda^{(j-1)}\|_2=1}& & (\x^{(j)})^T\z+\gamma\z_{sc}\lambda^T\begin{bmatrix}A &\v\end{bmatrix}\begin{bmatrix}\z\\\sigma\end{bmatrix}/\z_{sc}- \gamma r+f_k\nonumber \\
\mbox{subject to} && \z\in \left [0,2/\sqrt{n}\right ]^n, \z_{sc}=\sqrt{\|\z\|_2^2+\sigma^2} \nonumber \\
 & & \|\z^{(j)}\|_2^2=\hat{c}_{2,z}^{(j)}\nonumber \\
& & (\x^{(j-1)})^T\z^{(j)}=\hat{s}^{(j)} \nonumber \\
&& \z^{(j)}\in \left [0,2/\sqrt{n}\right ]^n, \z_{sc}^{(j)}=\sqrt{\|\z^{(j)}\|_2^2+\sigma^2},\label{eq:kit1}
\end{eqnarray}
where
\begin{equation}\label{eq:kit2}
f_k=\sum_{j=1}^{k}\left (\gamma_{j-1}\z_{sc}^{(j)}(\lambda^{(j-1)})^T\begin{bmatrix}A &\v\end{bmatrix}\begin{bmatrix}\z^{(j)}\\\sigma\end{bmatrix}/\z_{sc}^{(j)}- \gamma_{j-1} r\right ),
\end{equation}
and for completeness we also introduced $\hat{c}_{2,z}^{(1)}=\hat{c}_{1,z}^{(1)}$.

\vspace{.1in}
\noindent \xmyboxc{\bl{\emph{\textbf{2. Second step -- \dgr{Forming the Random dual} }}}}

\vspace{.1in}

To introduce the random dual we consider the following analogues to (\ref{eq:avoidsecclup1a2a5})-(\ref{eq:avoidsecclup1a2a7}). First we look at the following object
\begin{equation}\label{eq:kit3}
  f_{RD}= (\x^{(1)})^T\z+ \gamma \z_{sc} f_{RD,k+1}- \gamma r+\sum_{j=1}^{k}\gamma_{j-1}\z_{sc}^{(1)} f_{RD,j}- \gamma_{j-1} r,\nonumber \\
\end{equation}
where
\begin{equation}
f_{RD,k+1}  =  \lambda^T\g^{(k,q)}+(\h^{(k,p)})^T\z/\z_{sc} +h_0^{(k,p)}\sigma/\z_{sc}, \label{eq:kit4}
\end{equation}
and
\begin{eqnarray}
f_{RD,j} & = & (\lambda^{(j-1)})^T\g^{(j-1,q)}+(\h^{(j-1,p)})^T\z^{(j)}/\z_{sc}^{(j)} +h_0^{(j-1,p)}\sigma/\z_{sc}^{(j)}, \label{eq:kit5}
\end{eqnarray}
and for each $i$
\begin{eqnarray}\label{eq:kit6}
Q^{(k+1)} & = &   \mE\begin{bmatrix}\g_i^{(0,q)} & \g_i^{(1,q)} & \dots & \g_i^{(k,q)}\end{bmatrix}^T\begin{bmatrix}\g_i^{(0,q)} & \g_i^{(1,q)} & \dots & \g_i^{(k,q)}\end{bmatrix} \nonumber \\
P^{(k+1)} & = &     \mE\begin{bmatrix}\h_i^{(0,p)} & \h_i^{(1,p)} & \dots & \h_i^{(k,p)}\end{bmatrix}^T\begin{bmatrix}\h_i^{(0,p)} & \h_i^{(1,p)} & \dots & \h_i^{(k,p)}\end{bmatrix}.
\end{eqnarray}
$h_0$ doesn't really play much of role but one can for the completeness assume that it is an extension of $\h$ indexed by $0$ so that formally (\ref{eq:kit6}) holds for $h_0$ as well. Also, as expected, the components of all $\g$ and $\h$ are i.i.d. standard normals (the independence is over index $i$; also, $\g$ and $\h$ are independent of each other for any set of indices). Now one can define
\begin{equation}\label{eq:kit8}
  Z^{(k+1)}=\begin{bmatrix} \begin{bmatrix}\z^{(1)}\\\sigma\end{bmatrix}\frac{1}{\z_{sc}^{(1)}} & \begin{bmatrix}\z^{(2)}\\\sigma\end{bmatrix}\frac{1}{\z_{sc}^{(2)}}  & \dots & \begin{bmatrix}\z^{(k)}\\\sigma\end{bmatrix}\frac{1}{\z_{sc}^{(k)}}  & \begin{bmatrix}\z\\\sigma\end{bmatrix}\frac{1}{\z_{sc}} \end{bmatrix}.
\end{equation}
and
\begin{equation}\label{eq:kit8}
  \Lambda^{(k+1)}=\begin{bmatrix} \lambda^{(0)} & \lambda^{(1)} & \dots & \lambda^{(k-1)} & \lambda \end{bmatrix}.
\end{equation}
One then has for the  \bl{\textbf{random dual}} that corresponds to (\ref{eq:avoidsecclup1a2a7}) the following
\begin{eqnarray}
\min_{Q^{(k+1)}}\max_{P^{(k+1)}}\max_{\gamma,\gamma_j\geq 0}\min_{\z,\z^{(j)}} \max_{\|\lambda\|_2=1,\|\lambda_j\|_2=1}& & f_{RD} \nonumber \\
\mbox{subject to} && \z\in \left [0,2/\sqrt{n}\right ]^n, \z_{sc}=\sqrt{\|\z\|_2^2+\sigma^2} \nonumber \\
 & & \|\z^{(j)}\|_2^2=\hat{c}_{2,z}^{(j)}\nonumber \\
& & (\x^{(j-1)})^T\z^{(j)}=\hat{s}^{(j)} \nonumber \\
&& \z^{(j)}\in \left [0,2/\sqrt{n}\right ]^n, \z_{sc}^{(j)}=\sqrt{\|\z^{(j)}\|_2^2+\sigma^2}\nonumber \\
& &  (Z^{(k+1)})^TZ^{(k+1)}= Q^{(k+1)}\nonumber \\
& & (\Lambda^{(k+1)})^T\Lambda^{(k+1)}= P^{(k+1)}.\label{eq:kit9}
\end{eqnarray}
Where the remarks similar to (\ref{eq:avoidsecclup1a2a8}) and (\ref{eq:avoidsecclup1a2a9}) remain in place. In other words, the elements of matrices $Q^{(k+1)}$ and $P^{(k+1)}$ are basically the predicated concentrating points of all possible the so-called optimal achieving $\begin{bmatrix}\z\\\sigma\end{bmatrix}$- and $\lambda$-cross-overlaps, respectively.
 Similarly to what we had in the previous section, the concentrating values are not just over the standard randomness but also over certain the so-called Gibbsian measures randomness as well.

\vspace{.1in}
\noindent \xmyboxc{\bl{\emph{\textbf{3. Third step -- \dgr{Rehandling the Random dual of the first $k$ iterations} }}}}

One starts with handling the first iteration
\begin{eqnarray}
\min_{\z^{(1)}} \max_{\|\lambda_0\|_2=1}& & \z_{sc}^{(1)}f_{RD,1} \nonumber \\
\mbox{subject to}  & & \|\z^{(1)}\|_2^2=\hat{c}_{2,z}^{(1)}\nonumber \\
& & (\x^{(0)})^T\z^{(1)}=\hat{s}^{(1)} \nonumber \\
&& \z^{(1)}\in \left [0,2/\sqrt{n}\right ]^n, \z_{sc}^{(1)}=\sqrt{\|\z^{(1)}\|_2^2+\sigma^2},\label{eq:kit10}
\end{eqnarray}
then moves to the second and so on. The key output quantities after the $k$-th iteration (some of which are also needed for the $(k+1)$-th iteration) are
\begin{equation}
\x_i^{(j,s)},\z_i^{(j)},\lambda^{(j-1)}, 1\leq j\leq k,\label{eq:kit11}
\end{equation}
and
\begin{equation}\label{eq:kit12}
  \phi^{(k)}=\{\bl{p_{err}^{(k)},\hat{s}^{(k)},\hat{d}_2^{(k)},\hat{d}_1^{(k)}},\prp{\hat{\nu}^{(k)},\hat{\nu}_2^{(k)},\hat{\gamma}^{(k)},\hat{P}^{(k)},\hat{Q}^{(k)},\hat{c}_{2,z}^{(k)},\hat{s}_{2}^{(k)},\hat{s}_{3}^{(k)}}\},
\end{equation}
where we point out that $\hat{\nu}^{(k)}$ and $\hat{s}_{2}^{(k)}$ are $(k-1)$-dimensional vectors (this will become clearer below).

\vspace{.1in}
\noindent \xmyboxc{\bl{\emph{\textbf{4. Fourth step -- \dgr{Handling the real Random dual of the $(k+1)$-th iteration} }}}}

Following into the footsteps of what was done earlier, we finally have the following optimization problem (essentially the $(k+1)$-th iteration analogue to the second iteration's (\ref{eq:avoidsecclup1a2a15}))
\begin{eqnarray}
\min_{\z} \max_{\|\lambda\|_2=1}& & \z_{sc}f_{RD,k+1} \nonumber \\
\mbox{subject to}  & & \z\in \left [0,2/\sqrt{n}\right ]^n, \z_{sc}=\sqrt{\|\z\|_2^2+\sigma^2}\nonumber \\
& & \|\z\|_2^2=c_{2,z}\nonumber \\
& & (\x^{(j,s)})^T\z=s_{2,j},1\leq j\leq k \nonumber \\
& & \frac{1}{\sqrt{n}}\1^T\z=s_3 \nonumber \\
& &  (Z^{(k+1)})^TZ^{(k+1)}= Q^{(k+1)}\nonumber \\
& & (\Lambda^{(k+1)})^T\Lambda^{(k+1)}= P^{(k+1)}.\label{eq:kit13}
\end{eqnarray}
where the first $k$ columns of both $Z^{(k+1)}$ and $\Lambda^{(k+1)}$ are obtained after the $k$-th iteration. Analogously to (\ref{eq:avoidsecclup1a2a16}) one should here also keep in mind that
\begin{equation}\label{eq:kit14}
  Q_{k+1,j}^{(k+1)}=\begin{bmatrix}\z\\\sigma\end{bmatrix}^T\begin{bmatrix}\z^{(j)}\\\sigma\end{bmatrix}/\z_{sc}/\z_{sc}^{(j)}
  =\frac{s_3-s_{2,j}+\sigma^2}{\sqrt{c_{2,z}+\sigma^2}\sqrt{\hat{c}_{2,z}^{(j)}+\sigma^2}}.
\end{equation}
Taking $f_{RD,k+1}$ from (\ref{eq:kit4}) and plugging it back in (\ref{eq:kit13}) we have
\begin{eqnarray}
\min_{\z} \max_{\|\lambda\|_2=1}& & \z_{sc}(\lambda^T\g^{(k,q)}+(\h^{(k,p)})^T\z/\z_{sc} +h_0^{(k,p)}\sigma/\z_{sc}) \nonumber \\
\mbox{subject to}  & & \z\in \left [0,2/\sqrt{n}\right ]^n, \z_{sc}=\sqrt{\|\z\|_2^2+\sigma^2}\nonumber \\
& & \|\z\|_2^2=c_{2,z}\nonumber \\
& & (\x^{(j,s)})^T\z=s_{2,j},1\leq j\leq k \nonumber \\
& & \frac{1}{\sqrt{n}}\1^T\z=s_3 \nonumber \\
& &  (Z^{(k+1)})^TZ^{(k+1)}= Q^{(k+1)}\nonumber \\
& & (\Lambda^{(k+1)})^T\Lambda^{(k+1)}= P^{(k+1)}.\label{eq:kit15}
\end{eqnarray}
Neglecting the last term in the objective we finally have the following analogue to (\ref{eq:avoidsecclup1a2a17})
\begin{eqnarray}
\min_{\z} \max_{\|\lambda\|_2=1}& & \z_{sc}\lambda^T\g^{(k,q)}+(\h^{(k,p)})^T\z \nonumber \\
\mbox{subject to}  & & \z\in \left [0,2/\sqrt{n}\right ]^n, \z_{sc}=\sqrt{\|\z\|_2^2+\sigma^2}\nonumber \\
& & \|\z\|_2^2=c_{2,z}\nonumber \\
& & (\x^{(j,s)})^T\z=s_{2,j},1\leq j\leq k \nonumber \\
& & \frac{1}{\sqrt{n}}\1^T\z=s_3 \nonumber \\
& & (\Lambda^{(k+1)})^T\Lambda^{(k+1)}= P^{(k+1)}.\label{eq:kit16}
\end{eqnarray}
We will also denote
\begin{eqnarray}
f_{sph}^{(k+1)}= \frac{1}{\sqrt{m}}\mE\max_{\|\lambda\|_2=1}& & \lambda^T\g^{(k,q)}\nonumber \\
\mbox{subject to}  & &  (\Lambda^{(k+1)})^T\Lambda^{(k+1)}= P^{(k+1)},\label{eq:kit17}
\end{eqnarray}
and then rewrite (\ref{eq:kit16}) as
\begin{eqnarray}
\min_{\z} \max_{\|\lambda\|_2=1}& & \sqrt{m}\z_{sc}f_{sph}^{(k+1)}+(\h^{(k,p)})^T\z \nonumber \\
\mbox{subject to}  & & \z\in \left [0,2/\sqrt{n}\right ]^n, \z_{sc}=\sqrt{\|\z\|_2^2+\sigma^2}\nonumber \\
& & \|\z\|_2^2=c_{2,z}\nonumber \\
& & (\x^{(j,s)})^T\z=s_{2,j},1\leq j\leq k \nonumber \\
& & \frac{1}{\sqrt{n}}\1^T\z=s_3,\label{eq:kit18}
\end{eqnarray}
which for all practical purposes is an analogue to (\ref{eq:avoidsecclup1a2a19}). Now one can proceed as in the analysis of the second iteration right after (\ref{eq:avoidsecclup1a2a19}) and write the resulting Lagrange dual to obtain a problem structurally similar to (\ref{eq:avoidclup11})
\begin{eqnarray}
\max_{\gamma,\nu,\nu_2}\min_{\z} & & {\cal L}(\gamma,\nu,\nu_2)\nonumber \\
\mbox{subject to}  & & \z\in \left [0,2/\sqrt{n}\right ]^n,\label{eq:kit19}
\end{eqnarray}
where
\begin{eqnarray}\label{eq:kit20}
{\cal L}(\gamma,\nu,\nu_2) & = & \sqrt{\alpha n}\sqrt{\|\z\|_2^2+\sigma^2}f_{sph}^{(k+1)}\nonumber \\
& & +
\h^{(k,p)}\z +\gamma (\|\z\|_2^2-c_{2,z})+\sum_{j=1}^{k}\tilde{\nu}_j ((\x^{(j,s)})^T\z-s_{2,j})+\nu_2(\1^T\z-\sqrt{n}s_3).
\end{eqnarray}
Following closely what we did earlier, we use $\xi_{RD}^{(k)}(\alpha,\sigma;P^{(k+1)},Q^{(k+1)},c_{2,z},s_{2,j},s_3,\gamma,\tilde{\nu}_j,\nu_2)$ to denote the expected value of the above objective after it is scaled by $\sqrt{n}$. One can then follow further the machinery of say \cite{StojnicDiscPercp13}) and analogously to (\ref{eq:avoidseccclup12}) (and earlier (\ref{eq:avoidclup12})) define
\begin{eqnarray}
f_{box,k}(\h^{(k,p)};c_{2,z},s_{2,j},s_3) = \max_{\gamma,\tilde{\nu}_j,\nu_2}\min_{\z}&& \h^{(k,p)}\z +\gamma (\|\z\|_2^2-c_{2,z})+\sum_{j=1}^{k}\tilde{\nu}_j ((\x^{(j,s)})^T\z-s_{2,j})\nonumber \\
& &+\nu_2(\frac{1}{\sqrt{n}}\1^T\z-s_3)\nonumber \\
\mbox{subject to} & & \z\in \left [0,2/\sqrt{n}\right ]^n.\label{eq:kit21}
\end{eqnarray}
Now, here is the key point. When one compares (\ref{eq:avoidseccclup12}) to (\ref{eq:avoidclup12})
the difference is an extra constraint related to $s_3$. On the other hand, when one compares (\ref{eq:kit21}) to (\ref{eq:avoidseccclup12}) the difference is that instead of one constraint related to $s_2$ in (\ref{eq:avoidseccclup12}) one here has $k$ constraints related to $s_{2,j}$. However, the same conclusion made after (\ref{eq:avoidseccclup12}) applies here as well.
In other words, one can still utilize the solution obtained after (\ref{eq:avoidseccclup12}) with a few modifications to account for $s_{2,j}$ and $\tilde{\nu}_j$. One then effectively replaces (\ref{eq:avoidsecclup13}) (and earlier (\ref{eq:avoidclup13}) and ultimately (110) from \cite{StojnicDiscPercp13}) with
\begin{eqnarray}
f_{box,k+1}(\h^{(k,p)};c_{2,z},s_{2,j},s_3)  =  \max_{\gamma,\nu} & & \frac{1}{\sqrt{n}}\left (\sum_{i=1}^{n}f_{box,k+1}^{(1)}(\h_i^{(1,p)},\gamma,\tilde{\nu}_j,\nu_2)\right )\nonumber \\
&&-\sum_{j=1}^{k}\tilde{\nu}_j s_2\sqrt{n}-\nu_2 s_3\sqrt{n}-\gamma c_{2,z}\sqrt{n},\nonumber \\\label{eq:kit22}
\end{eqnarray}
where
\begin{equation}
f_{box,k+1}^{(1)}(\h_i^{(k,p)},\gamma,\tilde{\nu}_j,\nu_2)=\begin{cases}0, & \h_i^{(k,p)}+\sum_{j=1}^{k}\tilde{\nu}_j\x^{(j,s)}_i+\nu_2\geq 0\\
-\frac{(\h_i^{(k,p)}+\sum_{j=1}^{k}\tilde{\nu}_j\x^{(j,s)}_i+\nu_2)^2}{4\gamma}, & -4\gamma\leq \h_i^{(k,p)}+\sum_{j=1}^{k}\tilde{\nu}_j\x^{(j,s)}_i+\nu_2\leq 0\\
2(\h_i^{(k,p)}+\sum_{j=1}^{k}\tilde{\nu}_j\x^{(j,s)}_i+\nu_2)+4\gamma, & \h_i^{(1,p)}+\sum_{j=1}^{k}\tilde{\nu}_j\x^{(j,s)}_i+\nu_2\leq -4\gamma,
\end{cases}\label{eq:kit23}
\end{equation}
with the usual scaling discussion that we had after (\ref{eq:avoidsecclup14}) being applicable here to $\gamma$, $\tilde{\nu}_j$, $\nu_2$, and $\x^{(0)}_i$ as well. Analogously to (\ref{eq:avoidsecclup14a}) and (\ref{eq:avoidsecclup15a})  one then also has for the optimizing $\z_i$ and $\x_i^{(k+1,s)}$
\begin{eqnarray}
\z_i^{(k+1)} & = & \frac{1}{\sqrt{n}}\min \left (\max\left (0,-\left (\frac{\h_i^{(k,p)}+\sum_{j=1}^{k}\tilde{\nu}_j\x^{(j,s)}_i+\nu_2}{2\gamma}\right )\right ),2\right )\nonumber \\
\x_i^{(k+1,s)} & = & \frac{1}{\sqrt{n}}-\z_i^{(2)}=\frac{1}{\sqrt{n}}\left (1-\min \left (\max\left (0,-\left (\frac{\h_i^{(k,p)}+\sum_{j=1}^{k}\tilde{\nu}_j\x^{(j,s)}_i+\nu_2}{2\gamma}\right )\right ),2\right )\right ),\label{eq:kit24}
\end{eqnarray}
where $\x^{(j,s)}_i,1\leq j\leq k$ are obtained after the $k$-th iteration as stated in (\ref{eq:kit11}). Analogously to (\ref{eq:avoidsecclup15}) and (\ref{eq:avoidsecclup16}) we then have
\begin{equation}
\mE f_{box,k+1}^{(1)}(\h_i^{(k,p)},\gamma,\tilde{\nu}_j,\nu_2)=\rho I_{1}^{(k+1)}(\gamma,\nu,\nu_2,\hat{\nu}^{(1)})+(1-\rho) I_{1}^{(k+1)}(\gamma,\nu,\nu_2,-\hat{\nu}^{(1)}),\label{eq:kit25}
\end{equation}
where
\begin{equation}
I_{1}^{(k+1)}(\gamma,\nu,\nu_2,\hat{\nu}^{(1)})  =  \mE((\h_i^{(k,p)}+\sum_{j=1}^{k}\tilde{\nu}_j\x^{(j,s)}_i+\nu_2)\z_i^{(k+1)}+\gamma\left (\z_i^{(k+1)}\right )^2),
\label{eq:kit26}
\end{equation}
and for $\gamma<0$ the term under the expectation is basically zero if negative. Finally one arrives at
\begin{eqnarray}
\xi_{RD}^{(k+1)}(\alpha,\sigma;P^{(k+1)},Q^{(k+1)},c_{2,z},s_{2,j},s_3,\gamma,\tilde{\nu}_j,\nu_2) & = & \sqrt{\alpha}\sqrt{c_{2,z}+\sigma^2} f_{sph}^{(k+1)} +\mE f_{box,k+1}^{(1)}(\h_i^{(k,p)},\gamma,\tilde{\nu}_j,\nu_2)\nonumber \\
& &-\sum_{j=1}^{k}\tilde{\nu}_j s_{2,j}-\nu_2 s_3-\gamma c_{2,z}. \label{eq:kit27}
\end{eqnarray}
Below we present a brief summary of the above analysis. Since it conceptually closely follows the summaries that we presented after the analysis of the first and the second iteration we will try to make this summary as short as possible and basically rely on many ideas already introduced in earlier sections.

\subsection{CLuP -- summary of the $(k+1)$-th iteration performance analysis}
\label{sec:clupkitsummary}

We split the summary into two parts. The first one that is basically trivial and the second one that contains the key components of the analysis.

\vspace{.1in}
\noindent \xmyboxc{\bl{\emph{\textbf{Summarized formalism to handle the CLuP's $(k+1)$-th iteration }}}}

As mentioned above there are two parts that we recognize as critical in understanding the whole analysis mechanism.

\vspace{.1in}
\noindent \xmyboxc{\bl{\emph{\textbf{I) First part -- \dgr{Handling the first $k$ iterations} }}}}

This basically assumes just a simple recognition that the whole mechanism is in a way inductive in nature for any $k>2$. So, to start the induction one then assumes that the first $k$ iterations are doable (for $k=1$ and $k=2$ we have already shown that this is indeed the case) and continues further. To continue further one also recognizes the conclusion of the third step in the above discussion. That essentially amounts to recognizing that the key output quantities after the $k$-th iteration are
\begin{equation}
\x_i^{(j,s)},\z_i^{(j)},\lambda^{(j-1)}, 1\leq j\leq k,\label{eq:sumkit1}
\end{equation}
and
\begin{equation}\label{eq:sumkit2}
  \phi^{(k)}=\{\bl{p_{err}^{(k)},\hat{s}^{(k)},\hat{d}_2^{(k)},\hat{d}_1^{(k)}},\prp{\hat{\nu}^{(k)},\hat{\nu}_2^{(k)},\hat{\gamma}^{(k)},\hat{P}^{(k)},\hat{Q}^{(k)},\hat{c}_{2,z}^{(k)},\hat{s}_{2}^{(k)},\hat{s}_{3}^{(k)}}\},
\end{equation}
where $\hat{\nu}^{(k)}$ and $\hat{s}_{2}^{(k)}$ are $(k-1)$-dimensional vectors (it is obvious but for the completeness we also mention that $\hat{\nu}^{(k)}$ is essentially the vector of the optimal $\tilde{\nu}_j,1\leq j\leq k-1$ at the $k$-th iteration and analogously, $\hat{s}_{2}^{(k)}$ is the vector of the optimal $s_{2,j},1\leq j\leq k-1$ at the $k$-th iteration). As usual, a particular emphasis is on
\begin{eqnarray}\label{eq:sumkit3}
\bl{p_{err}^{(k)}} & - &  \mbox{probability of error after the $k$-th iteration}\nonumber \\
\bl{\hat{s}^{(k)}} & = &  \mE((\x^{(k-1)})^T\x^{(k,s)}) - \mbox{objective value after the $k$-th iteration} \nonumber \\
\bl{\hat{d}_2^{(k)}} & = & \mE\|\x^{(k,s)}\|_2^2 - \mbox{squared norm after the $k$-th iteration} \nonumber \\
\bl{\hat{d}_1^{(k)}} & = & \mE\x_{sol}^T\x^{(k,s)} - \mbox{inner product with $\x_{sol}$ after the $k$-th iteration}.
\end{eqnarray}

\vspace{.1in}
\noindent \xmyboxc{\bl{\emph{\textbf{II) Second part -- \dgr{Handling the $(k+1)$-th iteration} }}}}

We start with writing analogously to (\ref{eq:sumsecavoidclup17a}) and (\ref{eq:sumsecavoidclup17aa})
 \begin{eqnarray}
\phi_a^{(k+1)}=\mbox{arg} \min_{s_{2,j},s_3} & & \frac{s_{2,k}-\hat{d}_1^{(k)}}{\sqrt{\hat{d}_2^{(k)}}}\nonumber \\
\mbox{subject to} & & \min_{Q^{(k
+1)}}\max_{P^{(k+1)}}\min_{0\leq c_{2,z}\leq 4}\max_{\gamma,\tilde{\nu}_j,\nu_2}\xi_{RD}^{(k+1)}(\alpha,\sigma;P^{(k+1)},Q^{(k+1)},c_{2,z},s_{2,j},s_3,\gamma,\tilde{\nu}_j,\nu_2)=r,\nonumber\\\label{eq:sumkit4}
\end{eqnarray}
where
\begin{equation}\label{eq:sumkit5}
\phi_a^{(k+1)}=\{\hat{P}^{(k+1)},\hat{Q}^{(k+1)},\hat{\nu}^{(k+1)},\hat{\nu}_2^{(k+1)},\hat{\gamma}^{(k+1)},\hat{c}_{2,z}^{(k+1)},\hat{s}_2^{(k+1)},\hat{s}_3^{(k+1)}\},
\end{equation}
and obviously $\hat{\nu}^{(k+1)}$ is the $k$-dimensional vector of the optimal $\tilde{\nu}_j,1\leq j\leq k$ and $\hat{s}_{2}^{(k+1)}$ is the $k$-dimensional vector of the optimal $s_{2,j},1\leq j\leq k$.
One can then repeat all the steps between (\ref{eq:sumsecavoidclup17aa}) and (\ref{eq:sumsecavoidclup17a5}) to arrive at
\begin{eqnarray}
\phi_b^{(k+1)}=\mbox{arg} \min_{s,d_1^{(k+1)},d_2^{(k+1)},s_{2,j}} & & s\nonumber \\
\mbox{subject to} & & \max_{P^{(k+1)}}\min_{0\leq c_{2,z}\leq 4}\max_{\gamma,\nu,\nu_2}\xi_{RD}^{(k+1)}(\alpha,\sigma;P^{(k+1)},Q^{(k+1)},c_{2,z},s_{2,j},s_3,\gamma,\tilde{\nu}_j,\nu_2)=r\nonumber \\
& & s_{2,k}=\hat{d}_1^{(k)}+s\sqrt{\hat{d}_2^{(k)}}\nonumber \\
& & s_3=1-d_1^{(k+1)}\nonumber \\
& & c_{2,z}=d_2^{(k+1)}-2d_1^{(k+1)}+1\nonumber \\
& & Q_{k+1,j}^{(k+1)} =\frac{s_3-s_{2,j}+\sigma^2}{\sqrt{c_{2,z}+\sigma^2}\sqrt{\hat{c}_{2,z}^{(j)}+\sigma^2}},\label{eq:sumkit6}
\end{eqnarray}
where
\begin{equation}\label{eq:sumkit7}
\phi_b^{(k+1)}=\{\hat{P}^{(k+1)},\hat{Q}^{(k+1)},\hat{\nu}^{(k+1)},\hat{\nu}_2^{(k+1)},\hat{\gamma}^{(k+1)},\hat{s}^{(k+1)},\hat{d}_2^{(k+1)},\hat{d}_1^{(k+1)},\hat{s}_2^{(k+1)}\}.
\end{equation}
Following (\ref{eq:sumsecavoidclup17a6}) and (\ref{eq:sumsecavoidclup17a7}) one also has
\begin{equation}\label{eq:sumkit8}
  p_{err}^{(k+1)}=1-(\rho p_{cor}(\hat{\nu}^{(1)})+(1-\rho) p_{cor}(-\hat{\nu}^{(1)})),
\end{equation}
where
\begin{equation}\label{eq:sumkit9}
 p_{cor}(\hat{\nu}^{(1)})=\mE((\mbox{sign}(\x^{(k+1,s)})+1)/2).
\end{equation}
We skip rewriting the trivial analogues/adjustments of the considerations (\ref{eq:sumsecavoidclup17a8}) and (\ref{eq:sumsecavoidclup17a9}) and instead focus
at the output of the $(k+1)$-th iteration. Besides the solution, $\x^{(k+1,s)}$, the following set of \bl{critical} plus \prp{auxiliary} parameters is the output of the $(k+1)$-th iteration:
\begin{equation}\label{eq:sumkit10}
  \phi^{(k+1)}=\{\bl{p_{err}^{(k+1)},\hat{s}^{(k+1)},\hat{d}_2^{(k+1)},\hat{d}_1^{(k+1)}},\prp{\hat{\nu}^{(k+1)},\hat{\nu}_2^{(k+1)},\hat{\gamma}^{(k+1)},\hat{P}^{(k+1)},\hat{Q}^{(k+1)},\hat{c}_{2,z}^{(k+1)},\hat{s}_{2}^{(k+1)},\hat{s}_{3}^{(k+1)}}\},
\end{equation}
where we again for simplicity use the wording to emphasize
\begin{eqnarray}\label{eq:sumki11}
\bl{p_{err}^{(k+1)}} & - &  \mbox{probability of error after the $(k+1)$-th iteration}\nonumber \\
\bl{\hat{s}^{(k+1)}} & = &  \mE((\x^{(k)})^T\x^{(k+1,s)}) - \mbox{objective value after the $(k+1)$-th iteration} \nonumber \\
\bl{\hat{d}_2^{(k+1)}} & = & \mE\|\x^{(k+1,s)}\|_2^2 - \mbox{squared norm after the $(k+1)$-th iteration} \nonumber \\
\bl{\hat{d}_1^{(k+1)}} & = & \mE\x_{sol}^T\x^{(k+1,s)} - \mbox{inner product with $\x_{sol}$ after the $(k+1)$-th iteration}.
\end{eqnarray}

\subsection{Numerical results -- $(k+1)$-th iteration}
\label{sec:clupkitnum}

Looking carefully at the above analysis one quickly observes that in principle all the quantities of interest can be determined. However, there are quite a few of them that need to be optimized and there are quite a few numerical integrations that may have to be performed along the lines of such optimizations. In a separate paper we will present a systematic way to determine all these parameters. To avoid being sidetracked with such a large number of numerical considerations in the introductory paper where the goal is to present the key concepts behind the complexity analysis, we here provide estimates obtained in a simpler and much faster way. Namely, instead of systematically handling all of the above parameters, for the third and higher iterations we utilized the random dual itself. Given that the random dual is a much simpler program than the original primal one can run it on much larger dimensions. We have done so and manually estimated matrices $P$, $Q$, and vectors $d_1$ and $d_2$ throughout the process. An estimate for the $P$ and $Q$ matrices that we obtained is the following
\begin{equation}
  P^{(5)}=\begin{bmatrix}
             1 & .70 & .63 & .625  & .60\\
           .70 &  1 & .95 & .92  & .88\\
           .63 & .95 &  1 & .99  & .98 \\
           .625 & .92 & .99 &  1  & .996 \\
           .60 & .88 & .98 & .996 & 1
                     \end{bmatrix}
\end{equation}
\begin{equation}
  Q^{(5)}=\begin{bmatrix}
             1 & .825  & .69  & .65   & .63\\
           .825 &  1   & .955 & .92   & .9\\
           .69 & .955 &   1  & .995  & .99\\
           .65 & .92  & .995 &  1    & .999\\
           .63 & .9   & .99  & .999  & 1
                                \end{bmatrix}.
\end{equation}
Moreover, in Table \ref{tab:tabkit1} we give the estimated values for vectors $d_1$ and $d_2$.
\begin{table}[h]
\caption{Estimates for $d_1$ and $d_2$; $\alpha=0.8$; $r_{sc}=1.3$}\vspace{.1in}
\hspace{-0in}\centering
\footnotesize{
\begin{tabular}{||c||c|c|c|c|c||}\hline\hline
$k$  & $1$ & $2$ & $3$ & $4$ & $5$ \\ \hline\hline
$\hat{d}_1^{(k)}$  &  $\prp{\mathbf{0.7628  }}$ & $\prp{\mathbf{0.9340  }}$ & $\prp{\mathbf{0.9640  }}$ & $\prp{\mathbf{0.9663  }}$ & $\prp{\mathbf{0.9667  }}$ \\ \hline
$\hat{d}_2^{(k)}$  & $\prp{\mathbf{0.7009  }}$ & $\prp{\mathbf{0.9064  }}$ & $\prp{\mathbf{0.9420  }}$ & $\prp{\mathbf{0.9445  }}$ & $\prp{\mathbf{0.9450  }}$ \\ \hline\hline
\end{tabular}}
\label{tab:tabkit1}
\end{table}
In Table \ref{tab:tabkit2}, we complement these values for $\hat{d}_1^{(k)}$ and $\hat{d}_2^{(k)}$ with the estimated values for $p_{err}^{(k)}$ and $\hat{s}^{(k)}$ as well. Since all these rely on some manual estimates they are a little bit different from the values that can be obtained from a more precise systematic numerical analysis.
\begin{table}[h]
\caption{Change in $p_{err}^{(k)}$, $\hat{s}^{(k)}$, $\|\x^{(k,s)}\|_2^2$, and $(\x_{sol})^T\x^{(k,s)}$ as $k$ grows; $\alpha=0.8$; $r_{sc}=1.3$}
\vspace{.1in}
\hspace{-0in}\centering
\footnotesize{
\begin{tabular}{||c||c|c|c|c||}\hline\hline
$k$  & $p_{err}^{(k)}$ & $-\hat{s}^{(k)}$ & $\hat{d}_2^{(k)}=\|\x^{(k,s)}\|_2^2$ & $\hat{d}_1^{(k)}=(\x_{sol})^T\x^{(k,s)}$ \\ \hline\hline
$1$  & $\mathbf{0.04571  }$ & $\mathbf{0.1314  }$ & $\mathbf{0.7009  }$ & $\mathbf{0.7628  }$ \\ \hline
$2$  & $\mathbf{0.00651  }$ & $\mathbf{0.9117  }$ & $\mathbf{0.9064  }$ & $\mathbf{0.9340  }$ \\ \hline
$3$  & $\prp{\mathbf{0.00051  }}$ & $\prp{\mathbf{0.9658  }}$ & $\prp{\mathbf{0.9410  }}$ & $\prp{\mathbf{0.9640  }}$ \\ \hline
$4$  & $\prp{\mathbf{0.00024  }}$ & $\prp{\mathbf{0.9715  }}$ & $\prp{\mathbf{0.9445  }}$ & $\prp{\mathbf{0.9663  }}$ \\ \hline
$5$  & $\prp{\mathbf{0.00020  }}$ & $\prp{\mathbf{0.9720  }}$ & $\prp{\mathbf{0.9450  }}$ & $\prp{\mathbf{0.9667  }}$ \\ \hline\hline
limit & $\mathbf{0.00016}  $ & $\mathbf{0.9721}  $ & $\mathbf{0.9451}  $ & $\mathbf{0.9668}  $  \\ \hline\hline
\end{tabular}}
\label{tab:tabkit2}
\end{table}
Another thing that we should point out is that from (\ref{eq:kit17}) one has
\begin{eqnarray}
f_{sph}^{(k+1)}= \max_{\|\lambda\|_2=1}& & \lambda^T\g^{(k,q)}\nonumber \\
\mbox{subject to}  & &  (\Lambda^{(k+1)})^T\Lambda^{(k+1)}= P^{(k+1)}.\label{eq:numreskit1}
\end{eqnarray}
Given the statistics of $\g$ and $\h$ and their connection to $P^{(k+1)}$ and $Q^{(k+1)}$ from (\ref{eq:kit6}) one can through a little bit of work repose the above problem so that it becomes deterministic and basically a function of $P^{(k+1)}$ and $Q^{(k+1)}$ and then solve it either numerically or in some cases even in a closed form. However, we found further detailing the explanations of this procedure as unnecessary, since it turns out that the structure of the optimal matrices $P^{(k+1)}$ and $Q^{(k+1)}$ in the case that we consider here is such that the resulting value for $f_{sph}^{(k+1)}$ will already for $k=2$ be very close to $1$. Namely, as mentioned above, for $k=1$ one has from (\ref{eq:avoidsecclup1a2a18}) and (\ref{eq:avoidsecclup1a2a19})
\begin{eqnarray}
f_{sph}^{(2)}& = & \left (q^{(1)}p^{(1)}+\sqrt{1-(q^{(1)})^2}\sqrt{1-(p^{(1)})^2}\right )=\left (Q_{1,2}^{(2)}P_{1,2}^{(2)}+\sqrt{1-(Q_{1,2}^{(2)})^2}\sqrt{1-(P_{1,2}^{(2)})^2}\right ) =0.9834.\nonumber \\\label{eq:numreskit2}
\end{eqnarray}
For $k=2$ one doesn't even need to be as precise as above. Instead a trivial ad-hoc choice that in a way resembles the one that gives the above equation
\begin{eqnarray}
  & & A_q = Q_{1,3}^{(3)} \qquad \mbox{and} \qquad  B_q=\frac{Q_{2,3}^{(3)}-Q_{1,3}^{(3)}Q_{1,2}^{(3)}}{\sqrt{1-(Q_{1,2}^{(3)})^2}} \nonumber \\
  & &  A_p = P_{1,3}^{(3)}\qquad \mbox{and} \qquad  B_p=\frac{P_{2,3}^{(3)}-P_{1,3}^{(3)}P_{1,2}^{(3)}}{\sqrt{1-(P_{1,2}^{(3)})^2}},\label{eq:numreskit3}
\end{eqnarray}
gives
\begin{eqnarray}
C_q=\sqrt{1-A_q^2-B_q^2} \quad \mbox{and} \quad C_p=\sqrt{1-A_p^2-B_p^2},\label{eq:numreskit4}
\end{eqnarray}
and finally
\begin{eqnarray}
f_{sph}^{(3)}\geq A_q A_p+B_q B_p+C_q C_p=0.9967.\label{eq:numreskit5}
\end{eqnarray}

\subsubsection{Simulations -- $(k+1)$-th iteration}
\label{sec:clupkitsim}

As in earlier sections, we below provide a collection of results obtained through simulations. In Table \ref{tab:tabkit3} we show the CLuP's simulated performance over the first $6$ iterations for $\alpha=0.8$, $n=400$, and $r_{sc}=1.3$. As earlier, since $r_{plt}=.1226$ one easily has $\xi_{RD}^{(k)}=r=r_{sc}r_{plt}=.1594$.
\begin{table}[h]
\caption{Change in $p_{err}^{(k)}$, $\hat{s}^{(k)}$, $\|\x^{(k,s)}\|_2^2$, and $(\x_{sol})^T\x^{(k,s)}$ as $k$ grows; $\alpha=0.8$; $r_{sc}=1.3$; $n=400$}\vspace{.1in}
\hspace{-0in}\centering
\footnotesize{
\begin{tabular}{||c||c|c|c|c||}\hline\hline
$k$  & $p_{err}^{(k)}$ & $-\hat{s}^{(k)}$ & $\hat{d}_2^{(k)}=\|\x^{(k,s)}\|_2^2$ & $\hat{d}_1^{(k)}=(\x_{sol})^T\x^{(k,s)}$ \\ \hline\hline
$1$  &$0.04828  $ & $0.13163  $ & $0.7005  $ & $0.7596  $ \\ \hline
$2$  &$0.01110  $ & $0.90970  $ & $0.8980  $ & $0.9239  $ \\ \hline
$3$  &$0.00235  $ & $0.95915  $ & $0.9328  $ & $0.9560  $ \\ \hline
$4$  &$0.00067  $ & $0.96843  $ & $0.9406  $ & $0.9633  $ \\ \hline
$5$  &$0.00029  $ & $0.97041  $ & $0.9423  $ & $0.9648  $ \\ \hline
$\bl{\mathbf{6}}$  &$\bl{\mathbf{0.00019}}  $ & $\bl{\mathbf{0.97084}}  $ & $\bl{\mathbf{0.9427}}  $ & $\bl{\mathbf{0.9652}}  $ \\ \hline\hline
limit & $\mathbf{0.00016}  $ & $\mathbf{0.9721}  $ & $\mathbf{0.9451}  $ & $\mathbf{0.9668}  $  \\ \hline\hline
\end{tabular}}
\label{tab:tabkit3}
\end{table}
One observes that already for a very small number of iterations (basically just $6$) the simulated performance approaches the theoretical one (which actually allows for any number of iterations). Finally, in Table \ref{tab:tabkit4} we show how these results compare to the above discussed random duality theory predictions not only on the ultimate limiting level but also through a much more challenging per iteration level. We again observe a very strong agreement.
\begin{table}[h]
\caption{Change in $p_{err}^{(k)}$, $\hat{s}^{(k)}$, $\|\x^{(k,s)}\|_2^2$, and $(\x_{sol})^T\x^{(k,s)}$ as $k$ grows; $\alpha=0.8$; $r_{sc}=1.3$; $n=400$; \textbf{theory--computed}/\prp{\textbf{theory--estimated}}/\bl{\textbf{simulated}}}\vspace{.1in}
\hspace{-0in}\centering
\footnotesize{
\begin{tabular}{||c||c|c|c|c||}\hline\hline
$k$  & $p_{err}^{(k)}$ & $-\hat{s}^{(k)}$ & $\hat{d}_2^{(k)}=\|\x^{(k,s)}\|_2^2$ & $\hat{d}_1^{(k)}=(\x_{sol})^T\x^{(k,s)}$ \\ \hline\hline
$1$  & $\bl{\mathbf{0.04828  }}/\mathbf{0.04571  }$ & $\bl{\mathbf{0.1316  }}/\mathbf{0.1314  }$ & $\bl{\mathbf{0.7005  }}/\mathbf{0.7009  }$ & $\bl{\mathbf{0.75963  }}/\mathbf{0.7628  }$ \\ \hline
$2$  & $\bl{\mathbf{0.01110  }}/\mathbf{0.00651  }$ & $\bl{\mathbf{0.9097  }}/\mathbf{0.9117  }$ & $\bl{\mathbf{0.8980  }}/\mathbf{0.9064  }$ & $\bl{\mathbf{0.92387  }}/\mathbf{0.9340  }$ \\ \hline
$3$  & $\bl{\mathbf{0.00235  }}/\prp{\mathbf{0.00051  }}$ & $\bl{\mathbf{0.9591  }}/\prp{\mathbf{0.9658  }}$ & $\bl{\mathbf{0.9328  }}/\prp{\mathbf{0.9420  }}$ & $\bl{\mathbf{0.95601  }}/\prp{\mathbf{0.9640  }}$ \\ \hline
$4$  & $\bl{\mathbf{0.00067  }}/\prp{\mathbf{0.00024  }}$ & $\bl{\mathbf{0.9684  }}/\prp{\mathbf{0.9715  }}$ & $\bl{\mathbf{0.9406  }}/\prp{\mathbf{0.9445  }}$ & $\bl{\mathbf{0.96325  }}/\prp{\mathbf{0.9663  }}$ \\ \hline
$5$  & $\bl{\mathbf{0.00029  }}/\prp{\mathbf{0.00020  }}$ & $\bl{\mathbf{0.9704  }}/\prp{\mathbf{0.9720  }}$ & $\bl{\mathbf{0.9423  }}/\prp{\mathbf{0.9450  }}$ & $\bl{\mathbf{0.96483  }}/\prp{\mathbf{0.9667  }}$ \\ \hline\hline
limit & $\mathbf{0.00016}  $ & $\mathbf{0.9721}  $ & $\mathbf{0.9451}  $ & $\mathbf{0.9668}  $  \\ \hline\hline
\end{tabular}}
\label{tab:tabkit4}
\end{table}
Also, for the matrices $P$ and $Q$ we obtained through the simulations the following
\begin{equation}
  P^{(5)}=\begin{bmatrix}
    1.0000 &   0.7101  &  0.6351  &  0.6146 &   0.6088   \\
    0.7101 &   1.0000  &  0.9387  &  0.8888 &   0.8682   \\
    0.6351 &   0.9387  &  1.0000  &  0.9849 &   0.9711   \\
    0.6146 &   0.8888  &  0.9849  &  1.0000 &   0.9965   \\
    0.6088 &   0.8682  &  0.9711  &  0.9965 &   1.0000   \\
                                \end{bmatrix}.
\end{equation}
and
\begin{equation}
  Q^{(5)}=\begin{bmatrix}
    1.0000 &   0.8352  &  0.7072 &   0.6501  &  0.6281 \\
    0.8352 &   1.0000  &  0.9447 &   0.8898  &  0.8636 \\
    0.7072 &   0.9447  &  1.0000 &   0.9844  &  0.9685 \\
    0.6501 &   0.8898  &  0.9844 &   1.0000  &  0.9962 \\
    0.6281 &   0.8636  &  0.9685 &   0.9962  &  1.0000 \\
                                \end{bmatrix}.
\end{equation}
Given a rather small problem size $n=400$ the agreement with the theoretical predictions is again very good.

\subsubsection{Simulations -- changing $r$ effect}
\label{sec:clupkitsimchanger}

 In Table \ref{tab:tabkit3} we showed the CLuP's simulated performance for $r_{sc}=1.3$. It turned out that only $6$ iterations were enough to get very close to the ultimate theoretical predictions (such a prediction does not impose the limit on the number of iterations). In Table \ref{tab:tabkit5} we recall on the results from Table \ref{tab:tabkeyclupcmpl} and show the CLuP's simulated performance for $r_{sc}=1.5$. To achieve a bit better concentrations we chose $n=800$ (here as well as everywhere where we discussed the numerical results, all quantities are assumed as averaged, i.e. when we for example write in Table \ref{tab:tabkit5} $\|\x^{(k,s)}\|_2^2$ what it really means is $\mE\|\x^{(k,s)}\|_2^2$; given the overwhelming concentrations the two are basically the same thing). Looking at the results shown in the table, one observes that already after $10$ iterations all key parameters achieve values that are almost identical to the theoretical predictions (the difference is basically on the \bl{\textbf{fifth decimal}}). We do also emphasize that while we increased $n$ from $400$ to $800$, it is still a fairly small number and it is rather remarkable how powerful/exact RDT is when it comes to providing performance characterization estimates.
\begin{table}[h]
\caption{Change in $p_{err}^{(k)}$, $\hat{s}^{(k)}$, $\|\x^{(k,s)}\|_2^2$, and $(\x_{sol})^T\x^{(k,s)}$ as $k$ grows; $\alpha=0.8$; $r_{sc}=1.5$; $n=800$}\vspace{.1in}
\hspace{-0in}\centering
\footnotesize{
\begin{tabular}{||c||c|c|c|c||}\hline\hline
$k$  & $p_{err}^{(k)}$ & $-\hat{s}^{(k)}$ & $\hat{d}_2^{(k)}=\|\x^{(k,s)}\|_2^2$ & $\hat{d}_1^{(k)}=(\x_{sol})^T\x^{(k,s)}$ \\ \hline\hline
 $1$   &$0.079723  $ & $  0.17683  $ & $  0.68174  $ & $  0.71573 $ \\ \hline
 $2$   &$0.036420  $ & $  0.89707  $ & $  0.88133  $ & $  0.88448 $ \\ \hline
 $3$   &$0.014820  $ & $  0.95550  $ & $  0.94091  $ & $  0.94706 $ \\ \hline
 $4$   &$0.004763  $ & $  0.97799  $ & $  0.96899  $ & $  0.97579 $ \\ \hline
 $5$   &$0.001240  $ & $  0.98763  $ & $  0.97984  $ & $  0.98639 $ \\ \hline
 $6$   &$0.000317  $ & $  0.99090  $ & $  0.98309  $ & $  0.98946 $ \\ \hline
 $7$   &$0.000136  $ & $  0.99178  $ & $  0.98395  $ & $  0.99026 $ \\ \hline
 $8$   &$0.000083  $ & $  0.99202  $ & $  0.98419  $ & $  0.99048 $ \\ \hline
 $9$   &$0.000063  $ & $  0.99208  $ & $  0.98425  $ & $  0.99054 $ \\ \hline
$\bl{\mathbf{10}}$   &$\bl{\mathbf{0.000060}}  $ & $  \bl{\mathbf{0.99210}}  $ & $  \bl{\mathbf{0.98427}}  $ & $  \bl{\mathbf{0.99055}} $ \\ \hline\hline
\textbf{limit -- theory} & $\mathbf{0.000053}  $ & $\mathbf{0.99211}  $ & $\mathbf{0.98428}  $ & $\mathbf{0.99056}  $  \\ \hline\hline
\end{tabular}}
\label{tab:tabkit5}
\end{table}
Finally, one can also observe that when $r_{sc}=1.5$ the number of iterations increased. That is indeed the property of the CLuP algorithm. In other words, the increase in the number of iterations is not because $n$ increased from $400$ to $800$ but rather because of the increase in $r$. In fact, actually quite opposite is true. Namely, as $n$ increases the number of iterations goes down. For example, our results for $n=400$ indicate that $6$ iterations should be enough. However, it is very likely that for $n$ large enough $5$ or possibly even $4$ iterations would suffice to get an excellent performance. In Table \ref{tab:tabkit6} we actually show the simulated results obtained for $n=800$ and all other parameters as in the scenario that corresponds to Table \ref{tab:tabkit4}.
\begin{table}[h]
\caption{Change in $p_{err}^{(k)}$, $\hat{s}^{(k)}$, $\|\x^{(k,s)}\|_2^2$, and $(\x_{sol})^T\x^{(k,s)}$ as $k$ grows; $\alpha=0.8$; $r_{sc}=1.3$; $n=800$; \textbf{theory--computed}/\prp{\textbf{theory--estimated}}/\bl{\textbf{simulated}}}\vspace{.1in}
\hspace{-0in}\centering
\footnotesize{
\begin{tabular}{||c||c|c|c|c||}\hline\hline
$k$  & $p_{err}^{(k)}$ & $-\hat{s}^{(k)}$ & $\hat{d}_2^{(k)}=\|\x^{(k,s)}\|_2^2$ & $\hat{d}_1^{(k)}=(\x_{sol})^T\x^{(k,s)}$ \\ \hline\hline
$1$  & $\bl{\mathbf{0.04470  }}/\mathbf{0.04571  }$ & $\bl{\mathbf{0.1294  }}/\mathbf{0.1314  }$ & $\bl{\mathbf{0.7046  }}/\mathbf{0.7009  }$ & $\bl{\mathbf{0.7658  }}/\mathbf{0.7628  }$ \\ \hline
$2$  & $\bl{\mathbf{0.00790  }}/\mathbf{0.00651  }$ & $\bl{\mathbf{0.9126  }}/\mathbf{0.9117  }$ & $\bl{\mathbf{0.9048  }}/\mathbf{0.9064  }$ & $\bl{\mathbf{0.9317  }}/\mathbf{0.9340  }$ \\ \hline
$3$  & $\bl{\mathbf{0.00121  }}/\prp{\mathbf{0.00051  }}$ & $\bl{\mathbf{0.9618  }}/\prp{\mathbf{0.9658  }}$ & $\bl{\mathbf{0.9374  }}/\prp{\mathbf{0.9420  }}$ & $\bl{\mathbf{0.9603  }}/\prp{\mathbf{0.9640  }}$ \\ \hline
$4$  & $\bl{\mathbf{0.00033  }}/\prp{\mathbf{0.00024  }}$ & $\bl{\mathbf{0.9705  }}/\prp{\mathbf{0.9715  }}$ & $\bl{\mathbf{0.9438  }}/\prp{\mathbf{0.9445  }}$ & $\bl{\mathbf{0.9658  }}/\prp{\mathbf{0.9663  }}$ \\ \hline
$5$  & $\bl{\mathbf{0.00020  }}/\prp{\mathbf{0.00020  }}$ & $\bl{\mathbf{0.9719  }}/\prp{\mathbf{0.9720  }}$ & $\bl{\mathbf{0.9449  }}/\prp{\mathbf{0.9450  }}$ & $\bl{\mathbf{0.9668  }}/\prp{\mathbf{0.9667  }}$ \\ \hline\hline
limit & $\mathbf{0.00016}  $ & $\mathbf{0.9721}  $ & $\mathbf{0.9451}  $ & $\mathbf{0.9668}  $  \\ \hline\hline
\end{tabular}}
\label{tab:tabkit6}
\end{table}
Comparing Table \ref{tab:tabkit6} to Table \ref{tab:tabkit4} one can now see that when $n=400$ CLuP achieves in $6$ iterations roughly the same level of precision that it achieves in $5$ iterations when $n=800$. Finally, we also mention that from Figure \ref{fig:fignum2} one can also observe that when $n=800$ with just $10$ iterations the CLuP is already super close not only to its own optimum but also to the ML.

\vspace{-.15in}
\section{Discussion}
\label{sec:disc}
\vspace{-.03in}

Before reaching the conclusion of the paper there are a couple of things that we would like to particularly emphasize.

\vspace{.1in}
\noindent \xmyboxc{\bl{\emph{\textbf{I) \dgr{Full characterization of each iteration versus the number of iterations scaling behavior} }}}}

Namely, the analysis that we presented above is substantially different from typical complexity analyses. The typical complexity analyses usually try to determine an estimate for the number of iterations before the algorithm can terminate. Also, more often than not, such estimates do not insist on full exactness and instead focus on the so-called scaling behavior. In the above analysis we chose a different but way more general (and much harder) approach that essentially circumvents both of the above points (it doesn't count how many iterations will be run before the algorithm terminates and it doesn't use the ``estimating scalings" type of analysis). Instead, we determine all the key performance characterization parameters for any iteration and then simply observe in what iteration they come close to the limiting values. This is of course, what ideally should be the ultimate goal of any performance characterization analysis, but due to its hardness it is rarely (if ever) possible. Instead it is often then approximated through the above ``determining the scaling behavior of the number of iterations" approach.

\vspace{.1in}
\noindent \xmyboxc{\bl{\emph{\textbf{II) \dgr{CLuP's structural simplicity -- the ultimate goal} }}}}

Besides the above mentioned very small number of iterations that characterizes the discussed CLuP's behavior, one should also observe that the CLuP's overall structural simplicity is rather significant as well. Keeping this in mind and recalling on the two main goals when it comes to the MIMO ML (\textbf{exact solution} and \textbf{low complexity}) one arrives at the following so to say MIMO ML grand-challenge:

\tcbset{colback=yellow!95!white,colframe=blue!95!white,fonttitle=\bfseries}
\begin{tcolorbox}[title=MIMO ML -- ultimate grand-challenge]
\begin{itemize}
  \item \textbf{\emph{Can one design an algorithm that is structurally simpler than CLuP, runs in a smaller number of iterations, and achieves (or beats) the \prp{\underline{exact}} MIMO ML in polynomial time?}}
\end{itemize}.\vspace{-.3in}
\end{tcolorbox}
Of course, at this point it seems rather inconceivable how one can hope to design a structure simpler and iterations wise faster than CLuP. Still many miracles are possible within mathematics, so one should not rule out option that one day this kind of miracle turns out to be possible as well.

\section{Conclusion}
\label{sec:conc}

In our companion paper \cite{Stojnicclupint19} we introduced a very powerful and relatively simple polynomial method for achieving the \textbf{exact} ML-detection performance at the receiving end of MIMO systems. We referred to the method and all technical concepts that are behind it as the Controlled Loosening-up (CLuP). Already through  the introduction of the method in \cite{Stojnicclupint19} it was relatively easy to see that it has a large number of rather remarkable features, achieved through an incredibly simple underlying structure. Since all of these features have fairly deep mathematical roots that are behind them, in the prototype paper \cite{Stojnicclupint19} we only provided a set of so to say first glance observations and left a sequence of more thorough discussions for a series of companion papers. This paper is one of such papers and deals with a particularly important feature of the CLuP algorithm, the computational complexity.

As already mentioned in \cite{Stojnicclupint19}, one of the very best features of the CLuP algorithm is its computational complexity. Since CLuP is an iterative procedure with structurally fairly simple iterations, the main contributing factor to its overall complexity becomes the number of running iterations. What was also observed in \cite{Stojnicclupint19}, was that the number of iterations is not only acceptable (say polynomial) but it actually fairly often behaves way better than that. Namely, it turned out that it is an incredibly small number that often in the regimes of interest is not even going above $10$ no matter what the problem dimensions are. In this paper we provided a careful analysis of the algorithm's behavior through the iterations. The analysis that we provided is different from the standard complexity analyses that typically determine just the overall number of iterations and the complexity per iteration. Instead of going through such a rather standard route we actually approached it in a substantially more general way and analyzed each iteration of the algorithm separately. As a result of such an approach we determined a full characterization of all algorithm's critical parameters for any of the iterations.

The presented analysis is deeply rooted in some of the most crucial concepts within the Random Duality Theory (RDT). In a connection with RDT, we first presented the analysis on the level of the algorithm's first iteration and then showed how one can transfer from the first to the second iteration. This is of course the key point as we were then able to show how this particular iteration transfer can be utilized in all later transfers from $k$-th to $(k+1)$-th iteration (for any $k>2$). We also presented a large set of results obtained through numerical simulations and observed a strong agreement between what the theory and simulations predict.

We do mention that since this is the companion paper of the introductory one \cite{Stojnicclupint19}, we here (as in \cite{Stojnicclupint19}) focus on the CLuP itself. However, as was already discussed in \cite{Stojnicclupint19}, the presented concepts are very general and go way beyond CLuP. We will in several additional companion papers discuss how the analysis presented here changes when one faces a large class of other problems solvable through algorithms that conceptually rely on CLuP.

\begin{singlespace}
\bibliographystyle{plain}
\bibliography{clupcmplRefs}
\end{singlespace}

\end{document}